\documentclass{theoretics}

\usepackage{xspace}
\usepackage[all]{xy}
\usepackage{mathtools}

\RequirePackage{ifthen}
\RequirePackage{amssymb}

\renewcommand{\iff}{\Leftrightarrow}

\newcommand{\Nat}{{\mathbb N}}
\newcommand{\Int}{{\mathbb Z}}
\newcommand{\Z}{{\mathbb Z}}
\newcommand{\Rat}{{\mathbb Q}}

\newcommand{\D}{{\mathbb D}}

\newcommand{\set}[1]{\{#1\}}
\newcommand{\xor}{\div}

\newcommand{\Aa}{{\mathcal A}}
\newcommand{\Bb}{{\mathcal B}}

\newcommand{\Ll}{{\mathcal L}}

\newcommand{\Ww}{{\mathcal W}}

\newcommand{\ign}[1]{}

\newcommand{\myunderbrace}[2]{\underbrace{#1}_{\mathclap{\txt{\scriptsize #2}}}}

\newcommand{\atoms}{\mathbb A}
\DeclareMathOperator{\lin}{\mathsf{Lin}}

\DeclareMathOperator{\spanvec}{\mathsf{Span}}
\DeclareMathOperator{\length}{\mathsf{length}}

\newcommand{\field}{\mathbb F}

\newcommand{\red}[1]{#1}
\newcommand{\blue}[1]{#1}

\newcommand{\schutz}{Sch\"utzenberger\xspace}
\newcommand{\poly}[1]{\mathrm{poly}(#1)}
\newcommand{\exptime}{{\sc{ExpTime}}\xspace }
\newcommand{\twoexpspace}{{\sc{2ExpSpace}}\xspace }
\newcommand{\twoexptime}{{\sc{2ExpTime}}\xspace }
\newcommand{\expspace}{{\sc{ExpSpace}}\xspace }
\newcommand{\pspace}{{\sc{PSpace}}\xspace }

\newcommand{\finsup}{\stackrel{\mathrm{fs}}{\to}}
\newcommand{\forms}[1]{#1 \finsup \field}
\newcommand{\formscrammed}[1]{#1 \! \finsup \! \field}

\ThCSyear{2024}
\ThCSarticlenum{13}
\ThCSreceived{Apr 18, 2023}
\ThCSaccepted{Feb 15, 2024}
\ThCSpublished{May 6, 2024} 
\ThCSdoicreatedtrue

\ThCSkeywords{set with atoms, orbit-finite set, orbit-finite-dimensional space, weighted automaton, register automaton}
\ThCSthanks{An extended version of the LICS 2021 paper~\cite{lics21}. Mikołaj Boja\'nczyk was supported by the Polish National Science Centre (NCN) grant “Polynomial
finite state computation” (2022/46/A/ST6/00072). Joanna Fijalkow was supported by the project BOBR that has received funding from the European Research Council (ERC) under the European Union’s Horizon 2020 research and innovation programme (grant agreement No. 948057).}

\title{Orbit-Finite-Dimensional Vector Spaces  and Weighted Register Automata}
\ThCSauthor[1]{Miko{\l}aj Boja\'nczyk}{bojan@mimuw.edu.pl}[0000-0002-7758-1072]
\ThCSauthor[1,2]{Joanna Fijalkow}{joanna.fijalkow@gmail.com}[0000-0003-0945-0801]
\ThCSauthor[3]{Bartek Klin}{bartek.klin@cs.ox.ac.uk}[0000-0001-5793-7425]
\ThCSauthor[4]{Joshua Moerman}{joshua.moerman@ou.nl}[0000-0001-9819-8374]
\ThCSaffil[1]{University of Warsaw, Poland}
\ThCSaffil[2]{University of Bordeaux, CNRS, LaBRI, Talence, France}
\ThCSaffil[3]{University of Oxford, UK}
\ThCSaffil[4]{Open Universiteit, Heerlen, The Netherlands}

\ThCSshortnames{M. Boja\'nczyk, J. Fijalkow, B. Klin, J. Moerman}
\ThCSshorttitle{Orbit-Finite-Dimensional Vector Spaces  and Weighted Register Automata}

\begin{document}

\maketitle

\begin{abstract}
We develop a theory of vector spaces spanned by orbit-finite sets. Using this theory, we give a decision procedure for equivalence of weighted register automata,  which are the common generalization of weighted automata and register automata for infinite alphabets. The algorithm runs in exponential time, and in polynomial time for a fixed number of registers. As a special case, we can decide, with the same complexity, language equivalence for unambiguous register automata, which  improves previous results in three ways:
(a) we allow for order comparisons on atoms, and not just equality;
(b) the complexity is exponentially better; and (c) we allow automata with guessing. 
\end{abstract}

\section{Introduction}
Weighted automata over a field were introduced in~\cite{schutzenberger1961definition} by \schutz. Such an automaton is defined in the same way as a nondeterministic automaton, with a set $Q$ of states and an input alphabet $\Sigma$, except that instead of having subsets for transitions, initial and final states, the automaton has weight functions into the underlying field:
\begin{align*}
    \myunderbrace{I \colon Q \to \field}{initial} 
    \quad 
    \myunderbrace{\delta \colon Q \times \Sigma \times Q \to \field}{transition}
    \quad 
    \myunderbrace{F \colon Q \to \field}{final}.
\end{align*}
The weight of a run is obtained by multiplying the weights along all transitions, the initial weight of the first state, and the final weight of the last state. The  automaton recognizes a weighted language, which  is the function $L \colon \Sigma^* \to \field$ that maps a word to the sum of weights of all runs on that word.   

\schutz proved that weighted automata can be minimized~\cite[Sec.~B]{schutzenberger1961definition},
which provides a polynomial time algorithm for equivalence.
In contrast, it is undecidable whether some field element $a \in \field$ is achieved as the weight of some  word,  already in the special case of weighted automata that are probabilistic~\cite[Thm.~22]{nasu1969mappings}. Equivalence  is undecidable for weighted automata over semirings that are not fields, e.g.~for the min-plus semiring~\cite[Cor.~4.3]{krob1992}.

One application of weighted automata (see~\cite{balle2015learning} for other ones) is a polynomial time algorithm for  language equivalence of unambiguous automata, i.e.~nondeterministic finite automata that have at most one accepting run for every  word. The algorithm is a simple reduction to equivalence of weighted automata:  a nondeterministic automaton can be viewed as a weighted automaton over the field of rational numbers, such that the weight of a  word is the number of accepting runs. For unambiguous automata, the number of accepting runs is either zero or one, and hence two unambiguous automata accept the same words if and only if the corresponding weighted automata are equivalent.

In this paper, we generalize weighted automata to infinite alphabets, motivated by the study of register automata, in particular the equivalence problem for unambiguous register automata~\cite{clementeUnamb2021, Colcombet15, MottetQ19Guessing, MottetQ19}.  The kinds of infinite alphabets that we study are constructed  using an infinite set $\atoms$ of atoms (also called data values) which can only be accessed in very limited ways; in the simplest case, they can only be compared for equality. Register automata are like finite  automata, except that they additionally use finitely many registers to store atoms that occurred in the  word.
They were introduced by Kaminski and Francez~\cite[Def.~1]{kaminskiFiniteMemoryAutomata1994}, under the name of \emph{finite memory automata}. This model has attracted much attention, and register automata are now one of the most widely studied  infinite state systems. Their decidability landscape is rather complex: for example, emptiness is decidable for nondeterministic register automata~\cite[Thm.~1]{kaminskiFiniteMemoryAutomata1994}, but universality is not~\cite[Thm~5.1]{nevenFiniteStateMachines2004}. More robust results can be achieved for deterministic register automata, but at a considerable loss of expressive power.

One of the numerous (unfortunately non-equivalent) variants of register automata are unambiguous register automata, for which  language equivalence was shown in~\cite[Thm.~1]{MottetQ19} to be  decidable in \twoexpspace, and in \expspace for a fixed number of registers. 
 These upper bounds were improved to \twoexptime and \exptime, respectively, in~\cite[Thm.~2]{clementeUnamb2021}. None of these proofs use weighted automata.

From the point of view of register automata, the main contributions of this paper are:
\begin{enumerate}
    \item We introduce a weighted version of register automata, and we prove that their equivalence problem can be solved in \exptime, and in polynomial time for a fixed number of registers.
    \item We show that our weighted automata have a robust theory, in particular they can be described in several different ways and admit canonical syntactic automata. 
    \item As an immediate application of our equivalence algorithm for weighted automata, we show that the language equivalence problem for unambiguous register automata can be solved in the same complexity, improving  exponentially on prior work~\cite{clementeUnamb2021, MottetQ19}.
    \item In~\cite{clementeUnamb2021, MottetQ19}, the  equivalence algorithms  work only  for automata which are non-guessing, in the sense that every pair (state, input letter) has finitely many outgoing transitions. Without the non-guessing assumption, decidability is only known in the case where one of the automata has just a single register~\cite[Thm.~10]{MottetQ19Guessing}. Our algorithm can be adapted so that it works for general unambiguous automata, without the non-guessing assumption, with the same complexity.
    \item Apart from atoms with equality only, our algorithm for weighted automata, and its applications to unambiguous automata, also work for atoms equipped with a total order. Previous algorithms for ordered atoms assume that  one of the automata has a single register~\cite[Thm.~2]{MottetQ19}.
\end{enumerate}

In our opinion, perhaps the most interesting aspect of this paper is not the applications described above, but the theory that is developed in order to obtain them: elements of linear algebra in the universe of sets with atoms~\cite{bojanczyk_slightly2018}, also known as nominal sets~\cite{pitts}.
Our central objects of study are vector spaces which are spanned by orbit-finite set of vectors. A typical example is the vector space $\lin (\atoms^k)$,  which consists of finite linear combinations of $k$-tuples of atoms. Such a vector space has two kinds of structure: we can take linear combinations of vectors, and we can apply atom automorphisms to them.  A natural concept -- which appears  in our equivalence algorithm for weighted automata -- is the  \emph{equivariant subspaces}, i.e.~subsets of the vector space  closed under both linear combinations and atom automorphisms. Our principal technical result is that for every $k$, there is a finite (exponential in $k$) bound on the maximal length of chains
\begin{align*}
\myunderbrace{V_1 \subsetneq V_2 \subsetneq \cdots \subsetneq V_n}{equivariant subspaces} \subseteq \lin (\atoms^k).
\end{align*}
This bound is enough to obtain an algorithm for checking equivalence on weighted automata. For other problems, such as automata minimization, we need to go beyond vectors spaces of the form $\lin (\atoms^k)$ and consider the more general class of spaces which are spanned by orbit-finite sets. Since the axiom of choice fails for sets with atoms (see~e.g.~\cite[Sec.~2.7]{pitts}), these orbit-finitely spanned spaces may not admit equivariant bases. On the other hand, they are closed under subspaces, Cartesian products, tensor products and dual spaces. We hope that these results can be used in future work, beyond applications to unambiguous automata. 

The structure of the paper is as follows. In Section~\ref{sec:atoms} we recall the basics of sets with atoms and orbit-finite sets, with particular focus on equality and ordered atoms. In Section~\ref{sec:wof-aut} we introduce weighted orbit-finite automata and state our main result about checking their equivalence, Theorem~\ref{thm:equivalence-orbit-finite-equality}. In Section~\ref{sec:length} we leave automata aside, and we study vector spaces of the form $\lin Q$, for orbit-finite sets $Q$. We prove, in Theorem~\ref{thm:ordered-length} and Corollary~\ref{cor:equality-length}, that these spaces have finite length for equality and ordered atoms. These results are used in Section~\ref{sec:equivalence} to provide an algorithm for checking equivalence, thereby proving Theorem~\ref{thm:equivalence-orbit-finite-equality}. In Section~\ref{sec:vector-spaces} we study the general notion of orbit-finitely spanned vector spaces and prove their numerous closure properties. This is used in Section~\ref{sec:minimisation} to introduce the class of orbit-finitely spanned automata, which have the same expressive power as weighted orbit-finite automata but admit minimization. As another application, in Section~\ref{sec:unamb} we develop a new algorithm for equivalence checking for unambiguous register automata.

This is an extended version of the conference paper~\cite{lics21}. In addition to full proofs of several results, including a new and vastly simplified proof of Theorem~\ref{thm:ordered-length}, it settles several questions left open in~\cite{lics21}. In particular, we give an exponential lower bound on length (Section~\ref{sec:lowerbound}), an example of an atom structure where the finite length property fails (Section~\ref{sec:vectoratoms}), and a proof that orbit-finitely spanned spaces over graph atoms are not closed under duals (Example~\ref{ex:counter-example-graph-atoms} in Section~\ref{sec:vector-spaces}).

\section{Orbit-finite sets}
\label{sec:atoms}

Our paper is based on the approach to register automata that uses \emph{sets with atoms}. The idea is to consider automata where all components are closed under atom automorphisms, and the states and input alphabets have finitely many elements up to atom automorphisms. This abstract view avoids cumbersome notation for representing states and transitions of register automata. We can also leverage existing results that treat atoms with more structure than equality only, e.g., with a total order.

 Our notation is based on~\cite{bojanczyk_slightly2018}. 
Fix a countably infinite relational structure $\atoms$, in the sense of model theory, i.e.~an underlying set equipped with some relations. Elements of this fixed structure will be called \emph{atoms}. In this paper we  mostly focus on two such structures:
\begin{itemize}
\item the {\em equality atoms} $\atoms=(\Nat,=)$, and
\item the {\em ordered atoms} $\atoms=(\Rat,<)$.
\end{itemize}
The main results of this paper, notably the finite length property in Section~\ref{sec:length}, are proved  for these atoms only, and we leave other atoms as future work.

An \emph{atom automorphism}  is any bijection of the underlying set of atoms which is consistent with all relations in $\atoms$. For the equality atoms this is any permutation of $\Nat$; for the ordered atoms, any order-preserving bijection on $\Rat$. 
In both settings, the set of all atoms is orbit-finite, which means that it has finitely many elements up to atom automorphisms. Other orbit-finite sets constructed with atoms include:
\begin{align}\label{eq:examples-of-orbit-finite}
\myunderbrace{\atoms^k}{$k$-tuples of atoms} 
\qquad  \qquad \qquad \myunderbrace{\atoms^{(k)}}{non-repeating \\
\scriptsize $k$-tuples of atoms} \qquad \qquad  \qquad
\myunderbrace{\textstyle{\atoms \choose k}}{sets of $k$ atoms}.
\end{align}

 To formally define ``sets constructed with atoms'' and the orbit-finite restriction, we use the cumulative hierarchy from set theory. (See~\cite[Sec.~3.1]{bojanczyk_slightly2018} or~\cite[Sec.~2.6]{pitts} for a more detailed presentation.)
The \emph{cumulative hierarchy (over atoms $\atoms$)}  is indexed by ordinal numbers, and defined as follows: on level 0 we find the atoms, and on the level indexed by an ordinal $\alpha>0$ we find the atoms and all sets whose elements are from levels $< \alpha$. For example, on level 1 we find every subset of the atoms.  Automorphisms $\pi : \atoms \to \atoms$ of the atoms act on sets in the cumulative hierarchy in the expected way. 

In the cumulative hierarchy one can encode data structures such as tuples, words, relations, functions etc., using standard set-theoretic machinery.
Automorphisms $\pi$ then act on tuples as expected: $\pi((x_1, \ldots, x_n)) = (\pi(x_1), \ldots, \pi(x_n))$.

For a tuple of atoms $\bar a \in \atoms^*$, an {\em $\bar a$-automorphism} is an atom automorphism $\pi$ such that $\pi(\bar a)=\bar a$. A tuple $\bar a$ is said to \emph{support} an element $x$ of the cumulative hierarchy if
$\pi(x) = x$ for every $\bar a$-automorphism $\pi$.
    We say that an element $x$ of the cumulative hierarchy is \emph{finitely supported} if it is supported by some tuple of atoms in $\atoms^*$. An element is called {\em equivariant} if it is supported by the empty tuple. An equivariant set may contain non-equivariant elements; e.g., the set of all atoms is equivariant, but its elements are not. 
    A \emph{set with atoms} is defined to be a set in the cumulative hierarchy which is finitely supported and with each of its members (and their members, and so on) finitely supported.

Sets with atoms provide a relaxed notion of finiteness, called orbit-finiteness.
We say that two sets with atoms $x$ and $y$ are in the same $\bar a$-orbit if there is some $\bar a$-automorphism $\pi$ such that $y = \pi(x)$. 
Being in the same $\bar a$-orbit is clearly an equivalence relation; its equivalence classes are called {\em $\bar a$-orbits}. A set with atoms $x$ is called \emph{orbit-finite} if there is some atom tuple $\bar a$ such that $x$ is a finite union of $\bar a$-orbits.

\begin{example}\label{ex:sets-with-atoms}
\begin{itemize}
\item Over equality atoms, the set $\atoms$ of atoms is orbit-finite, also for every $k \in \set{1,2,\ldots}$ the sets $\atoms^{(k)}$ and $\atoms\choose k$ from~\eqref{eq:examples-of-orbit-finite} are orbit-finite, with a single equivariant orbit each. The set $\atoms^k$ is also orbit-finite, with the number of equivariant orbits equal to the $k$-th Bell number, i.e., the number of equivalence relations on the set $\{1,\ldots,k\}$.
\item Over ordered atoms, the sets $\atoms$ and $\atoms\choose k$ have single equivariant orbit each. The set $\atoms^{(k)}$ is also orbit-finite, with $k!$ equivariant orbits; one of these orbits is the  increasing $k$-tuples of atoms. Each of these orbits is in equivariant bijection with $\atoms\choose k$.
\item Over ordered atoms, for every atom $a \in \atoms$ the set $\atoms - \set{a}$ has two $a$-orbits: the open intervals $(-\infty;a)$  and $(a;\infty)$.
\item The set $\atoms^*$ is not orbit-finite, neither for equality nor for ordered atoms, as words of different length are necessarily in distinct orbits.
\end{itemize}
\end{example}

An atom structure $\atoms$ is called \emph{oligomorphic} if for every $k$, the set $\atoms^k$ is orbit-finite. For oligomorphic atoms, orbit-finite sets behave well, e.g.~they are closed under finite products and finitely supported subsets~\cite[Lem.~3.24]{bojanczyk_slightly2018}. In this paper we only consider oligomorphic atoms; in particular, equality and ordered atoms are oligomorphic.

Orbit-finite sets can be represented in a finite way so that they can be used as inputs for algorithms. Two examples of such representations are set builder expressions~\cite[Sec.~4.1]{bojanczyk_slightly2018} or the $G$-set representation~\cite[Sec.~8]{BKL14}. The reader does not need to know these representations in detail; the important thing is that they support basic operations such as products, Boolean operations, or inclusion and membership checks.  

\section{Weighted orbit-finite automata}\label{sec:wof-aut}
We now introduce the main model for this paper, an orbit-finite generalization of weighted automata.
\begin{definition}[Weighted orbit-finite automaton]\label{def:weighted-orbit-finite} 
Fix an oligomorphic atom structure $\atoms$ and a field $\field$. 
    A \emph{weighted orbit-finite automaton} consists of orbit-finite sets $Q$ and $\Sigma$, called the \emph{states} and the {\em alphabet}, and finitely supported functions 
    \begin{align*}
        \myunderbrace{I : Q \to \field}{initial} 
        \quad 
        \myunderbrace{\delta : Q \times \Sigma \times Q \to \field}{transitions}
        \quad 
        \myunderbrace{F : Q \to \field}{final}.
    \end{align*}
    Furthermore, we always require the \emph{non-guessing condition:}
    \begin{itemize}
    \item[(*)]
        there are finitely many states with nonzero initial weight, and also for every state $q$ and input letter $a$, there are finitely many states $p$ such that the transition $(q,a,p)$ has nonzero weight.
    \end{itemize}
\end{definition}

Weights of runs and  words is defined in the same way as for the classical model of weighted automata with finitely many states. 
The non-guessing condition ensures that each word has only finitely many runs with non-zero weight; otherwise there could be difficulties in summing up the weights of infinitely many runs. Intuitively, the condition means that an automaton cannot spontaneously invent a fresh atom.

\begin{example}\label{ex:count-letters} Consider any oligomorphic atoms $\atoms$ and the field of rational numbers. We define a weighted orbit-finite automaton  which maps a word $w \in \atoms^*$ to the number of distinct atoms that appear in $w$. The states are
    \newcommand*{\theinitialstate}{\bot}
    \begin{align*}
      \myunderbrace{\set{ \theinitialstate }}{initial weight 1\\
    \scriptsize final weight 0} \qquad  + \qquad   \myunderbrace{\atoms}{initial weight 0\\
    \scriptsize final weight 1}
    \end{align*}
     for some $\bot\not\in\atoms$. The transition weight is $1$ for all triples
    \begin{align*}
    (\theinitialstate, a, \theinitialstate) \  (\theinitialstate, a, a) \ (a, b, a) \qquad \text{for } a \neq b \in \atoms,
    \end{align*}
    and $0$ for all other triples. 
    For every input word, all runs have weight $0$, except for the following runs, which have weight $1$: start in $\theinitialstate$, stay there until  the last occurrence of some atom $a$, and then stay in state $a$ until the end of the word. Since the number of runs with weight $1$ is the number of distinct atoms $a$ that appear in the word, the output of the automaton is the number of distinct atoms. 
    
This automaton is equivariant, in the sense that its state space and all three weight functions are equivariant.      
\end{example}

\begin{example} \label{ex:count-runs} The automaton in the previous example is a special case of a more general construction: counting accepting runs of a nondeterministic automaton. 
    Define a \emph{nondeterministic orbit-finite automaton} like a nondeterministic finite automaton, except that all components (the alphabet, states, transitions, initial and accepting sets) are required to be orbit-finite sets, see~\cite[Def.~5.7]{bojanczyk_slightly2018}. To such an automaton $\Aa$ one can associate a weighted orbit-finite automaton (over the field of rational numbers) as follows: the alphabet and states are the same, and:
\begin{itemize}
\item a state gets an initial weight $1$ if it was initial in $\Aa$, otherwise its initial weight is $0$;
\item a state gets an final weight $1$ if it was accepting in $\Aa$, otherwise its initial weight is $0$;
\item transition weight is $1$ for all transitions present in $\Aa$ and $0$ for the ones absent in $\Aa$.
\end{itemize}
This  weighted automaton maps a word to the number of accepting runs of $\Aa$. The construction makes sense only if $\Aa$ is non-guessing in the sense that there are finitely many inputs states, and for every pair (state, letter) there are finitely many outgoing transitions. 
\end{example}

As mentioned in Section~\ref{sec:atoms}, orbit-finite sets can be represented in a finite way. Therefore, weighted orbit-finite automata can also be represented in a finite way (assuming that field elements can be represented in a finite way).  This is because (a) a finitely  supported relation  on an orbit-finite set is also orbit-finite; and (b) a finitely supported function from an orbit-finite set to any set is itself an orbit-finite set.

\begin{example}\label{ex:register-automata}
A special case of weighted orbit-finite automata, where the finite representation is easier to see, is a  weighted $k$-register automaton. This is a weighted orbit-finite automaton where the input alphabet is  finitely many disjoint copies of the atoms, the states are finitely many copies of $(\atoms + \set{\bot})^k$, and all weight functions are equivariant. For equality and ordered atoms, the weight functions can be finitely represented using quantifier-free formulas, see~\cite[p.~6]{bojanczyk_slightly2018}; as a result, weighted register automata appear as a straightforward weighted version of register automata as introduced in~\cite{kaminskiFiniteMemoryAutomata1994}. We will focus on this special case in Section~\ref{sec:unamb}.
\end{example}

We now state the main result of this paper, which is an algorithm for checking the equivalence of weighted orbit-finite automata, 
assuming  that the atom structure is either the equality atoms or the ordered atoms. We do not know if the problem is decidable for other atom structures.

\begin{theorem}\label{thm:equivalence-orbit-finite-equality}  Assume that the atoms are the equality atoms $(\Nat, =)$ or the ordered atoms $(\Rat, <)$. The  equivalence problem for equivariant\footnote{This theorem would also work for finitely supported automata, but the notation for the complexity and orbit counts would be more involved.} weighted orbit-finite automata can be solved in deterministic time 
    \begin{align*}
    2^{\poly k} \cdot n^{O(k)}
    \end{align*}
    where 
    \begin{itemize}
        \item[$n$]  is the orbit count, i.e.~the number  of equivariant orbits in the disjoint union of the two  state sets;
        \item[$k$] is the atom dimension of the state spaces of the automata, i.e.~the smallest $k$ such that every state in both automata is supported by at most $k$ atoms.
        \end{itemize}
In particular, the equivalence problem is in \exptime, and polynomial time when the atom dimension~$k$ is fixed.
\end{theorem}

A lower bound for the problem is \pspace, which is the complexity of language equivalence of deterministic register automata. We do not know the exact complexity; it is worth pointing out that for weighted finite automata there is an equivalence algorithm in randomized polylogarithmic parallel time~\cite[Sec.~3.2]{kiefer2013} that uses the Isolating Lemma.

We now present a proof strategy for the theorem, which will be carried out in the next sections.  The first observation is that the equivalence problem reduces to the zeroness problem, which asks whether a single weighted orbit-finite automaton outputs zero for every  word. The reduction is as follows: given two equivariant weighted orbit-finite automata $\Aa_1$ and $\Aa_2$, we create a new weighted orbit-finite automaton $\Aa_1 - \Aa_2$, which is obtained by taking the disjoint union of $\Aa_1$ and $\Aa_2$, and flipping the  sign of the final weights in $\Aa_2$. The new automaton maps all  words to zero if and only if the original two automata were equivalent. Also, in the reduction, the atom dimension does not change, neither does the orbit count.

It remains to give an algorithm for zeroness of a single equivariant  weighted orbit-finite automaton, with states $Q$. Our proof follows the same lines as \schutz's algorithm. Write $\lin Q$ for the set of  finite linear combinations (over the field $\field$) of states, seen as a vector space with basis $Q$.  For an input   word $w$, define its \emph{configuration} to be the vector in $\lin Q$
which maps a state $q$ to the sum of pre-weights of all runs over $w$ that end in state $q$; the pre-weight of a run is defined by multiplying the initial weight of the first state and the weights of all transitions, without taking into account the final weights. Thanks to the non-guessing condition (*) in Definition~\ref{def:weighted-orbit-finite}, each configuration is indeed a finite linear combination, since there are finitely many runs with nonzero pre-weight.  Consider  the  chain
\begin{align*}
V_0 \subseteq V_1 \subseteq V_2 \subseteq \cdots \subseteq \lin Q
\end{align*}
where $V_i$ is the  subspace of $\lin Q$ that is spanned by the configurations of   words of length at most  $i$. (We do not intend to compute the subspaces in this chain, whatever that would mean;  the chain is only used in the analysis of the algorithm.) Because the automaton is equivariant, one can easily see that each subspace $V_n$ is also equivariant. 

In the finite-dimensional case studied by \schutz, where $Q$ was a finite set, we could conclude that the chain must stabilize in a number of steps that is bounded by the finite dimension of the vector space $\lin Q$. In the orbit-finite case, the vector space has infinite dimension, and thus it is not clear why the chain should stabilize in finitely many steps. Our main technical contribution is a proof that the chain does indeed stabilize, and furthermore the time to stabilize is consistent with the bounds in the statement of Theorem~\ref{thm:equivalence-orbit-finite-equality}.  This stabilization property is the subject of the next section, and in Section~\ref{sec:equivalence} we shall use the property to conclude the proof of Theorem~\ref{thm:equivalence-orbit-finite-equality}.

\section{Finite length property}
\label{sec:length}

A standard result in linear algebra says that a vector space has finite dimension if and only if it has finite length, where the  length is defined to be the longest chain of its subspaces.  Indeed, the length of a finite-dimensional space is equal to its dimension. Because of this easy correspondence the notion of length is seldom explicitly applied to vector spaces, and it becomes more important only in more general structures such as modules over a ring. However, the situation becomes more interesting for chains of {\em equivariant} subspaces. We define:

\begin{definition}[Length]\label{def:length} 
The {\em length} of an equivariant vector space $V$, denoted $\length(V)$, is the maximal length $n$ of a chain of proper inclusions on equivariant subspaces of $V$:
    \begin{align*}
    V_1 \subsetneq V_2 \subsetneq \cdots \subsetneq V_n\subsetneq  V.    
    \end{align*}
If a maximal length does not exist, we say that $V$ has infinite length.     An atom structure $\atoms$ has the \emph{finite length property} if for every number $k$, the vector space $\lin \atoms^k$ has finite length.
\end{definition}

The finite length property easily implies oligomorphicity. As we shall see in Section~\ref{sec:vectoratoms}, the converse implication does not hold. But first, let us consider the two basic atom structures. 

\subsection{Equality and ordered atoms}

The purpose of this section is to prove that both the equality atoms $(\Nat,=)$ and the ordered atoms $(\Rat, <)$ have the finite length property (Corollary~\ref{cor:equality-length} and Theorem~\ref{thm:ordered-length}, respectively). 
 Furthermore, the length of $\lin \atoms^k$ grows exponentially (and not worse) with $k$. The results of this section apply to an arbitrary field $\field$. 

 Definition~\ref{def:length} speaks of equivariant subspaces. Later, in Theorem~\ref{cor:finsup-length}, we will show that  for the equality and ordered atoms,   allowing a fixed non-empty support would make no difference.

\begin{example}\label{ex:moerman-space}
	Consider the equality atoms. The vector space $\lin \atoms$ has length 2, because it has exactly 3 equivariant subspaces which form a chain of two proper inclusions:
	\begin{align*}
	\set{0} \quad \subsetneq \quad V \quad \subsetneq \quad \lin \atoms,
	\end{align*}
	where $V$ is the subspace spanned by the set
	\[\set{a-b : a \neq b \in \atoms}.\]
	Equivalently, $V$ is the vector space of all vectors where all coefficients sum up to $0$.
	Let us prove that there are no other equivariant subspaces. Suppose that $W$ is an equivariant subspace which contains some nonzero vector 
	\begin{align}\label{eq:wvector}
	w = \lambda_1 a_1 + \cdots + \lambda_n a_n \quad \text{where $\lambda_i \in \field\setminus\set{0}$.}
	\end{align}
	Since $W$ is equivariant, it also contains the vector obtained from $w$ by replacing $a_n$ with some fresh atom $b_n$. By taking the difference of these vectors and dividing by $\lambda_n$, we see that $W$ contains  $a_n - b_n$. By equivariance, $W$ contains all vectors of the form $a - b$ for distinct atoms $a,b$, and thus $V \subseteq W$.

We now show that $W$ is either $V$ or the entire space $\lin \atoms$. Indeed, suppose that $W$ contains some $w$ as in~\eqref{eq:wvector} that is not in $V$.
		If $n > 1$ then we can subtract from $w$ the  vector $\lambda_1 \cdot (a_1  - a_2)\in V$, which results in another vector that is in $W$ but not in $V$, with a smaller $n$. By repeating this process, we see that $W$ contains a vector of the form $\lambda_1 a_1$ with $\lambda_1 \neq 0$, and hence, by equivariance, it is the entire space.
\end{example}

Definition~\ref{def:length} differs from the classical definition of length in that we only consider equivariant subspaces. Some classical properties of length  easily transfer to our case. One may even say that our definition is a special case of the classical one: an equivariant vector space can be seen as a module over the (non-commutative) group ring $\field[G]$, where $\field$ is the underlying field and $G$ is the automorphism group of $\atoms$. To keep the presentation elementary we do not pursue this correspondence, but we remark that all properties of length which are valid for modules over (non-commutative) rings, remain true for our definition.
For example, the following lemma has the same proof (see Appendix~\ref{app:SES}) 
as for the classical notion of length of a module (see~e.g.~\cite[Prop.~4.12]{goodearl_warfield}), and works for arbitrary oligomorphic atoms:
\begin{lemma}\label{lem:SES}
For any equivariant spaces $V\subseteq W$, and equivariant sets $P$, $Q$ with their disjoint union $P+Q$: \begin{itemize}
\item[(i)]	$\length(W)=\length(V)+\length(W/V)$;\footnote{Here $W/V$ denotes the quotient space as usual.}
\item[(ii)]	$\length(\lin(P+Q)) = \length(\lin P) + \length(\lin Q)$;
\item[(iii)] if there is an equivariant surjective function from $P$ to $Q$ then $\length(\lin Q)\leq\length(\lin P)$.
\end{itemize}
\end{lemma} 

From this we infer that in Definition~\ref{def:length}, we could have equivalently talked about arbitrary equivariant orbit-finite set, instead of just sets of the form $\atoms^k$:
\begin{corollary} For  atoms with the finite length property,  $\lin Q$ has finite length for every equivariant orbit-finite  $Q$.
\end{corollary}
\begin{proof}
	An orbit-finite set is a finite disjoint union of single-orbit sets, and for oligomorphic atoms every single-orbit set $Q$ is an image of a surjective function from $\atoms^k$~\cite{BKL14}, for $k$ the atom dimension of $Q$.
	Hence, we can use the closure properties from Lemma~\ref{lem:SES}.
\end{proof}

The following lemma is the key step in proving the finite length property for the equality and ordered atoms:

\begin{lemma}\label{lem:ordered-length-1o}
Consider the ordered atoms $\atoms=(\Rat,<)$. For every $k$, the length of $\lin{\atoms\choose k}$ is finite, and satisfies
\[
	\textstyle \length\big(\lin{\atoms\choose k}\big)
	\quad \leq \quad  1+k\cdot\length\big(\lin{\atoms\choose k-1}\big).
\] 
\end{lemma}
\begin{proof}
For any set $\alpha$ of $2k$ atoms:
\begin{align}\label{eq:alpha}
	\underbrace{a_1<b_1<\cdots<a_k<b_k}_\alpha
\end{align}
and for any $I\subseteq k = \{1,\ldots,k\}$, define $\alpha\rtimes I \in {\atoms\choose k}$ by:
\[
	\alpha\rtimes I = \{a_i \mid i\not\in I\}\cup \{b_i\mid i\in I\}.
\]
In words, $\alpha\rtimes I$ picks either $a_i$ or $b_i$ from $\alpha$ according to $I$.
Define the {\em cog} (on~$\alpha$), $\nu^\alpha\in \lin {\atoms\choose k}$, to be the vector:
\[
	\nu^\alpha = \sum_{I\subseteq k}(-1)^{|I|}(\alpha\rtimes I).
\]
For example, for $k=2$, the cog on $4$ atoms $\alpha=\{a_1,b_1,a_2,b_2\}$ ordered as in~\eqref{eq:alpha} is the vector:
\[
	\textstyle \nu^\alpha \quad=\quad \{a_1,a_2\}-\{a_1,b_2\}-\{b_1,a_2\}+\{b_1,b_2\} \quad\in\quad \lin{\atoms\choose 2}.
\]
Notice that, for a fixed $k$, all cogs form a single orbit in $\lin{\atoms\choose k}$.

\begin{claim}\label{clm:cogs-everywhere}
Every nontrivial equivariant subspace $V\subseteq \lin{\atoms\choose k}$ contains a cog.
\end{claim}
\begin{proof}
Pick any nonzero $v\in V$ and pick $\alpha=\{a_1,\ldots,a_k\}\in{\atoms\choose k}$ so that $v(\alpha)\neq 0$. Choose fresh atoms $b_1,\ldots,b_k$  to form a set as in~\eqref{eq:alpha}, so that no atom present in $v$ lies strictly between $a_i$ and~$b_i$ for any $i$. (This is possible because $\alpha$ and $v$ are both finite.) We say that an atom $c$ is {\em present} in $v$ if there is a set $\beta\in{\atoms\choose k}$ such that $v(\beta)\neq 0$ and $c\in\beta$.

For every $i=1,\ldots,k$, choose an atom automorphism $\pi_i$ such that:
\begin{itemize}
\item $\pi_i(a_i)=b_i$, 
\item $\pi_i(a_j)=a_j$ and $\pi_i(b_j)=b_j$ for $j\neq i$, and
\item $\pi_i(c)=c$ for all other atoms $c$ present in $v$.
\end{itemize}
This is possible thanks to our choice of the $b_i$'s. Then put $v_0=v$ and define  $v_1,\ldots,v_k\in V$ by induction:
\[
	v_i = v_{i-1} - \pi_i(v_{i-1}) 
\]
(note that $\pi_i(v_{i-1})\in V$ since $V$ is equivariant). It is easy to prove by induction on $i$ that: 
\begin{itemize}
\item $\big(\pi_i(v_{i-1})\big)(\alpha)=0$, hence
\item $v_{i}(\alpha) = v(\alpha)$;
\end{itemize}
in particular $v_k$ is nonzero.
Also by straightforward induction, each $v_i$ has the following properties for every $\beta\in{\atoms\choose k}$:
\begin{itemize}
\item $v_i(\beta)$ is nonzero only if for each $1\leq j\leq i$, $\beta$ contains either $a_j$ or  $b_j$ but not both,
\item if $\beta'$ arises from such $\beta$ by replacing $a_j$ with $b_j$ (or {\em vice versa}) for exactly one $j\leq i$, while keeping the other components unchanged, then $v_i(\beta)+v_i(\beta')=0$.
\end{itemize}
For $i=k$ this implies that $v_k$, divided by the scalar $v(\alpha)$, is a cog.
\end{proof}

Claim~\ref{clm:cogs-everywhere} implies that $\lin{\atoms\choose k}$ has a unique least nontrivial equivariant subspace: the one spanned by all cogs. We shall now give an explicit description of that subspace.

A vector $v\in\lin{\atoms\choose k}$ is called {\em balanced} if for every set $S\in{\atoms\choose{k-1}}$, and for every $S$-orbit $I\subseteq\atoms$ such that $I\cap S=\emptyset$:
\begin{align}\label{balanced}
	\textstyle \sum_{a\in I}v(S\cup\{a\})=0.
\end{align}
The above sum is formally infinite, but only finitely many summands in it are non-zero because~$v$ is a finite vector.
Note that if $S=\{a_1,\ldots,a_{k-1}\}$ where $a_1<\cdots<a_{k-1}$, then $S$-orbits disjoint from $S$ are exactly the $k$ open intervals:
\begin{align}\label{orbits}
	(-\infty,a_1),\ (a_1,a_2),\ldots,(a_{k-2},a_{k-1}),\ (a_{k-1},+\infty).
\end{align}
Balanced vectors form an equivariant subspace of $\lin{\atoms\choose k}$: if $v$ is balanced then $\pi(v)$ is balanced for every atom automorphism $\pi$, and the sum of two balanced vectors is balanced. We denote the space of balanced vectors by $B$. An immediate corollary of Claim~\ref{clm:cogs-everywhere} is that every cog is balanced.

\begin{claim}\label{clm:Bcogs}
$B$ is the subspace spanned by all cogs.
\end{claim}
\begin{proof}
It is enough to show that every balanced vector is a linear combination of cogs. So consider any nonzero vector $v\in B$. Let $\beta_1\sqsubseteq \beta_2\sqsubseteq\ldots\sqsubseteq\beta_M\in{\atoms\choose k}$ be the set of all sets that appear in $v$ with nonzero coefficients, ordered lexicographically. More precisely, we define $\beta\sqsubseteq\beta'$ if 
\begin{itemize}
\item $\min(\beta)<\min(\beta')$, or
\item $\min(\beta)=\min(\beta')=a$ and $\beta\setminus\{a\}\sqsubseteq\beta'\setminus\{a\}$.
\end{itemize}
Consider
\[
	\beta_{\textrm{min}} = \{a_1,\ldots,a_k\},
\]
the least set in $v$ in the lexicographic order. Assume $a_1<a_2<\cdots <a_k$. 

Since $v$ is balanced, there must be some $b_1\neq a_1$ with $b_1<a_2$ such that the set $\beta'=\{b_1,a_2,\ldots,a_k\}$ has a nonzero coefficient in $v$. (Consider $S=\{a_2,\ldots, a_{k}\}$ in \eqref{balanced}.) This means that $\beta_{\textrm{min}}\sqsubseteq\beta'$, and so $a_1<b_1$.

Symmetrically, there must be some $b_k\neq a_k$ with $b_k>a_{k-1}$ such that the set $\beta'=\{a_1,\ldots,a_{k-1},b_k\}$ has a nonzero coefficient in $v$. This means that $\beta_{\textrm{min}}\sqsubseteq\beta'$, and so $a_k<b_k$.

More generally, for every $1<i<n$ there must be some $b_i\neq a_i$ with $a_{i-1}<b_i<a_{i+1}$ such that the set $\beta_{\textrm{min}}$ with $a_i$ replaced by $b_i$ has a nonzero coefficient in $v$. Because $\beta_{\textrm{min}}$ is the least set in $v$ lexicographically, it must hold that $a_i<b_i$.

We obtain the following order of $2k$ atoms:
\[
	\underbrace{a_1<b_1<\cdots<a_k<b_k}_\alpha
\]
Note that every atom in this set is present in some set that appears in $v$ with a nonzero coefficient. 

Let $\nu^\alpha$ be the cog based on $\alpha$. Note that for every $\beta'$ with a nonzero coefficient in $\nu^\alpha$ we have $\beta_{\textrm{min}}\sqsubseteq \beta'$.

Let $\lambda\neq 0$ be the coefficient of $\beta_{\textrm{min}}$ in $v$.
Consider the vector
\[
	w = v-\lambda\cdot\nu^\alpha.
\]
This has the following crucial properties:
\begin{itemize}
\item every atom present in $w$ is also present already in $v$, i.e., no new atoms are introduced in~$w$,
\item for every $\beta'$ with a nonzero coefficient in $w$ we have $\beta_{\textrm{min}}\sqsubseteq \beta'$,
\item the coefficient of $\beta_{\textrm{min}}$ in $w$ is zero, i.e., $\beta_{\textrm{min}}$ does not appear in $w$.
\end{itemize}

By induction on the (finite) lexicographic order on $\atoms\choose k$ restricted  to atoms present in $v$, the vector $w$, and therefore also $v$, is a linear combination of cogs. This concludes the proof of Claim~\ref{clm:Bcogs}.
\end{proof}

Claims~\ref{clm:cogs-everywhere} and~\ref{clm:Bcogs} together imply that $B$ is the smallest nontrivial equivariant subspace of $\lin {\atoms\choose k}$. To finish the proof of Lemma~\ref{lem:ordered-length-1o}, for every $i\in\{1,\ldots,k\}$ consider the equivariant linear map:
\[
	\textstyle g_i:\lin{\atoms\choose k}\to\lin{\atoms\choose k-1}
\]
defined by:
\[
	\textstyle g_i(v)(S) = \sum_{a\in I}v(S\cup\{a\})
\]
for any $S\in{\atoms\choose{k-1}}$, where $I$ is the $i$'th orbit on the list~\eqref{orbits}. Tupling these functions for all $i$ we obtain an equivariant linear map:
\[
	\textstyle g:\lin{\atoms\choose k}\to\left(\lin{\atoms\choose k-1}\right)^k.
\]
By definition, the kernel of $g$ is $B$, so there is a subspace embedding:
\[
	\textstyle \left(\lin{\atoms\choose k}\right)/B \ \subseteq \ \left(\lin{\atoms\choose k-1}\right)^k.
\] 
Note that $B$ does not have any nontrivial equivariant subspaces, so $\length(B)=1$. 
By Lemma~\ref{lem:SES}(i) applied to $V=B$ and $W=\lin{\atoms\choose k}$ we thus obtain:
\begin{gather*}
	\textstyle \length\big(\lin{\atoms\choose k}\big) = 1+\length\left(\left(\lin{\atoms\choose k}\right)/B\right) \leq \\
	\textstyle \leq 1+\length\left(\left(\lin{\atoms\choose k-1}\right)^k\right) = 1+k\cdot\length\big(\lin{\atoms\choose k-1}\big)
\end{gather*}
and conclude the proof of Lemma~\ref{lem:ordered-length-1o}.
\end{proof}

The following is now easy:

\begin{theorem}\label{thm:ordered-length}
   The ordered atoms $\atoms=(\Rat,<)$ have the finite length property. For an equivariant orbit-finite set $Q$, the space $\lin Q$ has length at most
    \begin{align*}
   \text{(orbit count of $Q$)} \cdot (1+{\text{atom dimension of $Q$}})!
    \end{align*}
\end{theorem}
\begin{proof}
If $Q$ has only one orbit, then there is an equivariant surjective function from $\atoms\choose k$ to $Q$, where $k$ is the atom dimension of $Q$ (see~\cite[Lem.~3.20]{bojanczyk_slightly2018}). By induction on $k$, using Lemma~\ref{lem:ordered-length-1o}, we obtain
\begin{align}\label{eq:achoosebound}
	\textstyle \length\big(\lin{\atoms\choose k}\big) \leq (1+k)!
\end{align}
The theorem then follows by Lemma~\ref{lem:SES}(iii).
Finally, the case of multi-orbit $Q$ follows from Lemma~\ref{lem:SES}(ii).
\end{proof}

From this it is easy to deduce the finite length property for equality atoms. 

\begin{corollary}\label{cor:equality-length}
     The equality atoms $\atoms=(\Nat,=)$ have the finite length property. For an equivariant orbit-finite set $Q$, the space $\lin Q$ has length at most
    \begin{align*}
   \text{(orbit count of $Q$)} \cdot k!\cdot (1+k)!
    \end{align*}
where $k$ is the atom dimension of $Q$.
\end{corollary}
\begin{proof}
Extend  $\atoms$ with any total order isomorphic to that of the rational numbers. Any set $V$, equivariant over equality atoms, remains equivariant when considered as a set over ordered atoms along this order. In particular, if $V$ is a vector space, any chain of subspaces
\[
	V_1\subsetneq V_2 \subsetneq \cdots \subsetneq V
\]
equivariant over equality atoms, remains a chain over ordered atoms. As a result, the length of~$V$ over equality atoms does not exceed the length of $V$ over ordered atoms.

The set $\atoms^{(k)}$, a single-orbit set over equality atoms, when seen as an equivariant set over ordered atoms, has exactly $k!$ disjoint orbits. Each of these orbits is equivariantly (over ordered atoms) isomorphic to $\atoms\choose k$. Theorem~\ref{thm:ordered-length} then  implies that $\lin\big(\atoms^{(k)}\big)$, over equality atoms, has length at most $k!\cdot (1+k)!$.

Every single-orbit equivariant set $Q$ is an image of an equivariant surjective function from~$\atoms^{(k)}$, where $k$ is the atom dimension of $Q$; the lemma follows by Lemma~\ref{lem:SES}(iii). The case of multi-orbit $Q$ follows from Lemma~\ref{lem:SES}(ii).
\end{proof}

\subsection{Length of finitely supported chains}

So far, we have bounded the length of chains of equivariant subspaces. This will be enough to prove Theorem~\ref{thm:equivalence-orbit-finite-equality} for equivariant automata in Section~\ref{sec:equiv-alg}. However, for generalising that theorem to finitely supported automata a finite bound on the length of chains of finitely supported spaces is needed. Here it is important that the chains are uniformly supported, i.e., all subspaces in a chain are required to have the same support. (Uniformly supported chains have been considered before in nominal domain theory~\cite{TurnerWinskel09}.)

It is natural to conjecture that the finite length property (Definition~\ref{def:length}) implies a finite bound on the length of such chains, at least for the atom structures that guarantee the existence of least supports (see~\cite[Section 9]{BKL14}). However, we are only able to prove it for the equality and ordered atoms.

\begin{theorem}\label{cor:finsup-length}
	Consider the equality or ordered atoms. Let $\bar a$ be a tuple of atoms.  For every $\bar a$-supported orbit-finite set $Q$, there is a finite upper bound on the length of chains of $\bar a$-supported subspaces of $\lin Q$.
\end{theorem}
\begin{proof}
	Define the $\bar a$-length of a vector space to be the maximal length of an $\bar a$-supported chain of subspaces. In the special case when $\bar a$ is the empty tuple, we get the notion of length from Definition~\ref{def:length}. 
	Every $\bar a$-supported orbit-finite set of atom dimension $k$ can be obtained, using images under $\bar a$-supported functions and disjoint unions, from $\bar a$-orbits contained in $\atoms^k$. Since 
	 Lemma~\ref{lem:SES} holds, with the same proof,  for $\bar a$-length, it is enough to show that the $\bar a$-length is finite for $\lin Q$ when $Q$ is a single $\bar a$-orbit contained in $\atoms^k$. We now split into two proofs, depending on whether we deal with the equality or ordered atoms. 
	 \begin{itemize}
	 \item \emph{Equality atoms.} Choose some bijection 
	 \begin{align*}
	 f : \myunderbrace{(\atoms - \bar a)}{atoms that do  \\
	 \scriptsize not appear in $\bar a$} \to \atoms,
	 \end{align*}
	 which is possible since both sets are countably infinite. (Note that $f$ cannot be finitely supported.) Let $\ell \in \set{1,\ldots,k}$ be the number of coordinates that are not from $\bar a$ in some (equivalently, every) tuple from the $\bar a$-orbit $Q$.  We can lift $f$ to an injective function (in fact, a bijection) 
	 \begin{align*}
	 g : Q \to \atoms^{(\ell
	 )},
	 \end{align*}
	 which erases the coordinates that use atoms from $\bar a$ and applies $f$ to the remaining coordinates. One can easily see that 
	 \begin{align*}
	 \lin g : \lin Q \to \lin \atoms^{(\ell)}
	 \end{align*}
	  maps   $\bar a$-supported subspaces of $\lin Q$ to equivariant subspaces  of $\lin \atoms^{(\ell)}$, and preserves strict inclusions. Therefore the $\bar a$-length of $\lin Q$ is at most the (equivariant) length of $\lin \atoms^{(\ell)}$, and the latter length is finite by Corollary~\ref{cor:equality-length}.
	 \item \emph{Ordered atoms.} Let the atoms in $\bar a$ be 
	 \begin{align*}
	 a_1 < \cdots < a_n.
	 \end{align*}
	 For every $i \in \set{0,1,\ldots,n}$, consider the interval
	 \begin{align*}
	 \atoms_i = \set{ a : a_i < a < a_{i+1}},
	 \end{align*}
	 where $a_0$ is $-\infty$ and $a_{n+1}$ is $\infty$. Choose some order-preserving bijections 
	 \begin{align*}
	 f_i : \atoms_i \to \atoms.
	 \end{align*}
	 As in the previous item, let $\ell$ be the coordinates from $Q$ which avoid atoms from $\bar a$. 
	 By erasing the coordinates which use atoms from $\bar a$, and applying the appropriate functions~$f_i$ to the remaining coordinates, we get an injective function 
	 \begin{align*}
	 g : Q \to \atoms^{\ell}.
	 \end{align*}
	 Since the functions $f_i$ all have the same co-domain, namely $\atoms$, the image of the function will contain tuples that are not necessarily strictly increasing. Nevertheless, the linear lifting 
	 \begin{align*}
		g : \lin Q \to \lin \atoms^{\ell}
		\end{align*}
		maps $\bar a$-supported subspaces in $\lin Q$ to equivariant subspaces in $\lin  \atoms^\ell$, and hence by Theorem~\ref{thm:ordered-length} we obtain a finite bound on $\bar a$-supported chains in $\lin Q$. 
	 \end{itemize}
\end{proof}

\subsection{A lower bound on length}\label{sec:lowerbound}

Theorem~\ref{thm:ordered-length} and Corollary~\ref{cor:equality-length} give upper bounds on the length of chains of equivariant vector spaces. We will now show a lower bound which is not quite matching, but exponential in the atom dimension.

We will focus on the equality atoms, where the case of interest is $\lin(\atoms^{(k)})$ for a number $k$.

For $\alpha\in\atoms^{(k)}$ and $I\subseteq\{1,\ldots,k\}$, let $\alpha|_I$ denote the tuple $\alpha$ restricted to the coordinates from $I$, and define $B_I(\alpha) \subseteq \atoms^{(k)}$ to be the set of those tuples $\beta$ for which $\beta|_I=\alpha|_I$.
 
Given $I$, define the space $V_I\subseteq \lin(\atoms^{(k)})$ as the set of all vectors $v$ such that, for every $\alpha\in\atoms^{(k)}$:
\[
	\sum_{\beta\in B_I(\alpha)}v(\beta) \quad=\quad 0.
\]
For example, for $I=\emptyset$ the space $V_I$ is the space of those vectors where all coefficients add up to~$0$. For $I=\{1,\ldots,n\}$ we have $B_I(\alpha)=\{\alpha\}$ and the space $V_I$ is trivial.

If $I\subseteq J\subseteq \{1,\ldots,n\}$ then each set $B_I(\alpha)$ is a disjoint union of sets of the form $B_J(\beta)$ for some tuples $\beta$. As a result, $V_J\subseteq V_I$. Our main observation is a kind of converse to this:

\begin{lemma}\label{lem:VJVI}
For every $I,J_1,\ldots,J_m\subseteq\{1,\ldots,k\}$, if $I\not\subseteq J_i$ for every $i=1,\ldots,m$ then:
\[
	V_{J_1}\cap\cdots\cap V_{J_m}\not\subseteq V_I.
\]
\end{lemma}
\begin{proof}
Assume sets $I$ and $J_1,\ldots,J_m$ as above. Fix some pairwise distinct atoms $c_i$ for $i\not\in I$ and $a_i,b_i$ for $i\in I$. For every subset $H\subseteq I$, define the tuple $\alpha_H=(d_1,\ldots,d_k)\in\atoms^{(k)}$ by:
\[
	d_i = \left\{\begin{array}{cl}
		a_i & \text{if }i\in H \\
		b_i & \text{if }i\in I\setminus H \\
		c_i & \text{if }i\not\in I.
		\end{array}\right.
\]
Define the vector:
\[
	v = \sum_{H\subseteq I}(-1)^{|H|}\cdot \alpha_H.
\]
Then $v\in V_{J_i}$ for each $i=1,\ldots,m$, but $v\not\in V_I$.
\end{proof}

\begin{corollary}
The space $\lin(\atoms^{(k)})$ admits a chain of equivariant subspaces of length $2^k$.
\end{corollary}
\begin{proof}
Consider any enumeration of all subsets of $\{1,\ldots,k\}$:
\[
	J_1,J_2,\ldots,J_{2^k}
\]
such that $J_i\not\subseteq J_j$ for all $i>j$. This is easy to achieve: order the subsets by their size (increasingly), and subsets of equal size can be enumerated in any order. Then define spaces:
\[
	W_i = \bigcap_{j\leq i} V_{J_j}.
\]
Obviously $W_1\supseteq W_2\supseteq\cdots\supseteq W_{2^k}$, and by Lemma~\ref{lem:VJVI} all the inclusions are strict.
\end{proof}

\subsection{Failure of finite length: vector atoms}\label{sec:vectoratoms}

In this section we consider, as the structure of atoms, the countable infinite-dimensional vector space over $\Z_2$. We will show that these atoms do not have the finite length property. This is a rather confusing case in the context of this paper: in general we study vectors built of atoms, and now we wish atoms themselves to be vectors. We will therefore present the structure and its basic properties step by step.

Let $\D = \{a,b,c,d,\ldots\}$ be a countably infinite set of generators. We consider, as the set of atoms, 
\[
	\atoms = \lin(\D),
\]
the free vector space, over the field $\Z_2$, generated by $\D$. Vectors in this space (i.e., our atoms) are formal combinations of generators; because we work over $\Z_2$, these formal combinations are simply finite sets of generators, with addition defined by symmetric difference. We will represent atoms simply by listing the generators in them, with the zero vector denoted by $0$:
\[
	0, a, b, c, \ldots, ab, ac, \ldots, bce, \ldots
\]
and their addition will be denoted by $\oplus$, so that e.g.:
\[
	abc \oplus bcd = ad.
\]
To distinguish atoms from generators we will, in this section only, denote them by $u$, $v$, $w$ etc. 

The space $\atoms$ can be seen as a relational structure in a standard way, with one constant $0\in\atoms$ and one ternary relation
\[
	R(u,v,w)\ \iff\ u\oplus v=w.
\]
The automorphisms of this structure are exactly linear automorphisms of $\atoms$ seen as a vector space. It is a standard result that the structure $\atoms$ is oligomorphic.

For $v\in \atoms$ and $X\subseteq \atoms$, we write as usual: $v\oplus X = \{v\oplus w \mid w\in X\}$.

We will consider finite-dimensional subspaces of $\atoms$. For a number $n>0$, an $n$-dimensional subspace of $\atoms$ has exactly $2^n$ elements. One way to construct such a subspace is to take the subspace spanned by some chosen $n$ generators from $\D$. For example,
\[
\{0,a,b,c,ab,ac,bc,abc\}
\]
is a 3-dimensional subspace spanned by $a$, $b$ and $c$. Other subspaces exist; for example,
\[
\{0,ab,ac,bc\}
\]
is a 2-dimensional subspace.

For subsets $X,Y\subseteq\atoms$, we denote their symmetric difference (as sets) by $X\xor Y$. 

\begin{lemma}\label{fact:vec-basic}
For all $v\in\atoms$ and $X,Y,Z\subseteq\atoms$:
\begin{itemize}
\item[(i)] $X\cap(Y\xor Z) = (X\cap Y)\xor (X\cap Z)$,\label{eq:cap-diff}
\item[(ii)] $v\oplus(X\cap Y) = (v\oplus X)\cap (v\oplus Y)$,\label{eq:oplus-cap}
\item[(iii)] $v\oplus(X\xor Y) = (v\oplus X)\xor(v\oplus Y)$.\label{eq:oplus-diff}
\end{itemize}
\end{lemma}
\begin{proof}Routine exercise.
\end{proof}

We will be interested in symmetric differences of finite families of vector spaces of the same finite dimension.

\begin{lemma}\label{lem:nmxor}
For numbers $n>m>0$, let $V\subseteq\atoms$ be an $n$-dimensional vector space and let 
$V_1,V_2,\ldots,V_k$ be all its $m$-dimensional subspaces. Then
\begin{align}\label{eq:nmxor}
	V_1\xor V_2\xor\cdots\xor V_k = V.
\end{align}
\end{lemma}
\begin{proof}
It is a standard exercise to show that
\[
	k = \dfrac{(2^n-1)(2^n-2)\cdots(2^n-2^{m-1})}{(2^m-1)(2^m-2)\cdots(2^m-2^{m-1})}.
\]
Both the numerator and the denominator of this fraction are divisible by the same power of $2$; specifically it is $\frac{m(m-1)}{2}$. As a result, $k$ is an odd number. 

For any vector $v\in V$, let $\alpha(v)$ be the number of subspaces $V_i$ that contain $v$. Note that, due to the symmetry under automorphisms of $V$, for all non-zero $v,w\in V$ there is $\alpha(v)=\alpha(w)$; denote this number by $\alpha$.

Now calculate:
\[
	k\cdot2^m = \sum_{i=1}^k|V_i| = \sum_{v\in V}\alpha(v) = k+\sum_{v\in V\setminus\{0\}}\alpha(v) = k+(2^n-1)\cdot\alpha.
\] 
The number on the left is even and $k$ is odd, so $\alpha$ must be odd.

As a result, {\em every} vector in $V$ belongs to an odd number of subspaces $V_i$, which yields~\eqref{eq:nmxor}.
\end{proof}

The following basic property of finite-dimensional spaces over $\Z_2$ will be important for our purposes.

\begin{lemma}\label{lem:vsgap}
For every finite family $A_1,A_2,\ldots, A_k$ of $n$-dimensional vector subspaces of $\atoms$, the symmetric difference $A_1\xor A_2\xor\cdots\xor A_k$ is either empty or has size at least $2^n$.
\end{lemma}
\begin{proof}
See Appendix~\ref{app:vsgap}. 
\end{proof}

To see that Lemma~\ref{lem:vsgap} is not entirely obvious, let us have a look at some examples. First of all, the symmetric difference from the lemma may be empty even if all $A_i$ are distinct. For an example, consider $n=2$, $k=4$ and vector subspaces:
\begin{align*}
 V_1 &= \{0,a,b,ab\} & V_2 &= \{0,b,c,bc\} \\
 V_3 &= \{0,a,c,ac\} & V_4 &= \{0,ab,bc,ac\}. 
\end{align*}

The symmetric difference may also be of size strictly between $2^n$ and $2^{n+1}$, so in particular it may not be a vector subspace. For an example, consider $n=3$, $k=2$ and vector subspaces:
\begin{align*}
 V_1 &= \{0,a,b,c,ab,bc,ac,abc\} \\
 V_2 &= \{0,a,d,e,ad,de,ae,ade\};
\end{align*}
here $|V_1\xor V_2|=12$. Another example, for $n=2$ and $k=3$, is:
\[
 V_1 = \{0,a,b,ab\} \qquad V_2 = \{0,b,c,bc\} \qquad V_3 = \{0,c,d,cd\};
\]
here $|V_1\xor V_2\xor V_3|=6$. If we additionally consider
\[
 V_4 = \{0,ab,bc,ac\},
\]
the symmetric difference becomes $\{a,d,ac,cd\}$, showing that the symmetric difference may not be a vector subspace even if its size is $2^n$. 

\medskip

So far in this section we studied properties of the vector atoms, but have not considered vector spaces {\em over} those atoms yet. In fact we will need to consider only one such vector space and its equivariant subspaces: $\lin\atoms$, the free vector space generated by the set of atoms, over the field $\Z_2$. Vectors in $\lin\atoms$ are formal linear combinations of atoms, such as:
\begin{equation}\label{eq:vecvec-ex}
	0 + a + b + ab \qquad \text{or} \qquad a + bc + ac + bde.
\end{equation}
The addition symbol above is used to build formal combinations of atoms, and it should not be confused with $\oplus$, which is an operation to build atoms from other atoms. Thanks to the choice of the field, combinations as above can be seen simply as finite sets of atoms, with symmetric difference $\xor$ as addition.

The main result of this section is that $\lin\atoms$ has infinite length, so the vector atoms $\atoms$ does not have the finite length property.
\begin{theorem}
$\lin\atoms$ has an infinite chain of equivariant subspaces.
\end{theorem}
\begin{proof}
For any number $n$, let $W_n\subseteq\lin\atoms$ be the subspace spanned by those vectors in $\lin\atoms$ that, considered as subsets of $\atoms$, are $n$-dimensional vector subspaces of $\atoms$. (For example, in~\eqref{eq:vecvec-ex} the first vector is a subspace of $\atoms$, the second one is not.) It is easy to see that each space $W_n$ is equivariant.

For any $n>m$, the inclusion $W_n\subseteq W_m$ follows from Lemma~\ref{lem:nmxor}. Moreover, this inclusion is strict by Lemma~\ref{lem:vsgap}. Indeed, every nonzero vector in $W_n$ has at least $2^n$ elements, and $W_m$ contains vectors of smaller size.

As a result, there is an infinite chain of subspaces:
\[
	\cdots \subsetneq W_3\subsetneq W_2 \subsetneq W_1 \subsetneq \lin\atoms.
\]
\end{proof}

Note that the above infinite chain of subspaces is decreasing. We conjecture that for all oligomorphic atom structures $\atoms$, all spaces $\lin\atoms^k$ are {\em noetherian}, i.e. they do not admit {\em increasing} infinite chains of equivariant subspaces. 

\section{The equivalence algorithm}\label{sec:equiv-alg}
\label{sec:equivalence}

In this section, we use results from Section~\ref{sec:length} to complete the proof of  Theorem~\ref{thm:equivalence-orbit-finite-equality}. We assume that the atoms $\atoms$ are the equality or ordered atoms. 

\begin{lemma}\label{lem:short-distinguishing-inputs}
    Let $\Aa_1$ and $\Aa_2$ be equivariant weighted orbit-finite automata, and let $n$ be their orbit count and $k$ the atom dimension as in Theorem~\ref{thm:equivalence-orbit-finite-equality}. If the recognized weighted languages are different, then this difference is witnessed by some input word of length at most 
    \begin{align*}
         2^{\poly k} \cdot n.
        \end{align*}
\end{lemma}
\begin{proof}
    Let $Q$ be the state space of the difference automaton $\Aa_1 - \Aa_2$ as described in the reduction from equivalence to zeroness. 
     For an input word $w \in \Sigma^*$, let $[w] \in \lin Q$ be its corresponding configuration in the difference automaton, and let $V_i \subseteq \lin Q$ be the subspace that is spanned by configurations of  input words of length at most $i$.   
    By Theorem~\ref{thm:ordered-length}/Corollary~\ref{cor:equality-length}, we know that the chain $\set{V_i}_i$  must stabilize in a number of steps that is bounded as in the statement of the lemma.  This means that for every input word, its  configuration  is a  linear combination of configurations of short words; in particular the difference automaton can produce a nonzero output if and only if it can produce a nonzero output on a short input word.
\end{proof}

At this point, we could solve the equivalence problem by  guessing a short differentiating word, leading to a nondeterministic algorithm for non-equivalence, with running time as in the bound from Lemma~\ref{lem:short-distinguishing-inputs}. 
We can, however, improve this to get a deterministic algorithm, by using a reduction to the equivalence problem for finite weighted automata. 
Short input words necessarily use few atoms, and therefore the equivalence problem boils down to checking equivalence for input words that have few atoms. The latter problem is solved in the following lemma.

\begin{lemma}\label{lem:reduction-to-finite}
    Consider the following problem. 
\begin{itemize}
\item {\bf Input.} Two equivariant weighted orbit-finite automata $\Aa_1$ and $\Aa_2$ and a number $\ell \in \set{1,2,\ldots}$;
\item {\bf Question.} Are the two  automata equivalent on all input words supported by at most $\ell$ atoms?
\end{itemize}
This problem can be solved  in time polynomial in $n  \cdot \ell^k$, where the parameters $n$ and $k$ are defined as in Theorem~\ref{thm:equivalence-orbit-finite-equality}.
    \end{lemma}
\begin{proof}
    Choose a tuple $\bar a$ of $\ell$ atoms. Since the automata are equivariant, they are equivalent on inputs with at most $\ell$ atoms if and only if they are equivalent on inputs words supported by~$\bar a$. 

    By equivariance and condition (*) from Definition~\ref{def:weighted-orbit-finite}, if a state $q$  and an input letter $\sigma$ are both supported by $\bar a$, then the same is true for every state $p$ such that $(q,\sigma,p)$ is a transition with nonzero weight.  Therefore, when restricted to input words that are supported by $\bar a$, both automata only use states that are supported by $\bar a$.  These observations motivate the  following definition: for $i \in \set{1,2}$,   let  $\Aa_{i, \bar a}$ the  weighted automaton that is obtained from $\Aa_i$ by restricting  the states and alphabet to elements supported by $\bar a$, and restricting the transitions and weight functions to the new alphabet and states. 
    This automaton is finite, 
with size  bounded by $n \cdot \ell^k$. The lemma follows by applying the polynomial time algorithm for equivalence of finite weighted automata. 
\end{proof}

The above two lemmas complete the proof of Theorem~\ref{thm:equivalence-orbit-finite-equality}. The algorithm for zeroness is simply the algorithm from Lemma~\ref{lem:reduction-to-finite}, with $\ell$ being the bound  from Lemma~\ref{lem:short-distinguishing-inputs}.

Notice how our algorithm does not rely on an explicit computation of the chain of subspaces $\{V_i\}_i$ from the proof of Lemma~\ref{lem:short-distinguishing-inputs}. Doing so would require a method for calculating (a representation of) $V_{i+1}$ from (a representation of) $V_i$, and a method of checking whether two subspaces of $\lin Q$ are equal, to detect stabilisation. Variants of the latter problem, under the guise of solving orbit-finite systems of linear equations, are studied in~\cite{GHL22}.

\section{Vector spaces with atoms}
\label{sec:vector-spaces}
In this section, we study in more depth vector spaces that are spanned by orbit-finite sets.  Apart from their independent interest, these results will be used in Section~\ref{sec:minimisation} to minimize weighted automata, and in Section~\ref{sec:unamb} to decide equivalence for unambiguous automata.

Vector spaces of the form $\lin Q$, where $Q$ is an orbit-finite set, could be called {\em orbit-finite-dimensional}, since they have an orbit-finite basis. These spaces are general enough to treat weighted orbit-finite automata. Some other spaces admit orbit-finite bases too; in~\cite[Thm.~3.3]{GHL22}, this is shown for a variant of $\lin Q$ where vectors are infinite (but finitely supported) linear combinations, rather than finite ones. Nevertheless, vector spaces with orbit-finite bases are not very robust.  The reason is that extracting a basis uses choice, see~\cite[Thm.~1]{blassbases}, and the principle of choice fails for sets with atoms.

\begin{example}
    \label{ex:no-equivariant-basis}
    Recall the vector space  $\lin \atoms$ that was discussed in Example~\ref{ex:moerman-space}, and the subspace $V$ that was spanned by
    \begin{align*}
        X = \set{a-b : a \neq b \in \atoms}.
    \end{align*}
    The set $X$ is not a basis, since the vectors $a-b$ and $b-a$ (or $a-b$, $b-c$ and $a-c$) are linearly dependent. However, $X$ is a single equivariant orbit, so it does not have any nonempty proper subset that is equivariant. Therefore, no equivariant subset of $X$ is a basis. In fact, $V$ does not have any equivariant basis, even if we allow bases that are not contained in $X$. It does have, however, a finitely supported basis contained in $X$: for any fixed atom $a_0 \in \atoms$, the set
    \begin{align*}
     \set{a_0 - b : b \in \atoms - \set {a_0} }
    \end{align*} 
    is a basis.
\end{example}

In the previous example, there was no equivariant basis, but there was a finitely supported one. In the next example there is no finitely supported basis at all.
\begin{example}
    \label{ex:no-orbit-finite-basis}
    Consider the equality atoms and the space  $\lin(\atoms^{(2)})$ over the field of rationals. A vector  in this space can be visualized as a weighted directed finite graph, where 
vertices are atoms and the weight of an edge $(a,b)$ is the corresponding coefficient in the vector. 
Consider the subspace of $\lin(\atoms^{(2)})$  spanned by 
    \begin{align*}
    X = \set{(a,b) - (b,a) : a \neq b \in \atoms}.
    \end{align*}
This subspace consists of graphs where for every pair of vertices, the connecting edges in both directions have opposite weights. We claim that there is no finitely supported subset of $X$ that is a basis of this space. Indeed, suppose that $Y \subseteq X$ is a finitely supported basis. It is easy to see that for every two distinct atoms $a, b$, the set $Y$ must contain one of the  vectors
\[
	(a,b)-(b,a) \qquad \text{or} \qquad (b,a)-(a,b).
\]
If the atoms $a,b$ are fresh (i.e. they do not belong to the least support of $Y$), by swapping these two atoms we map one of the vectors to the other, hence {\em both} vectors must belong to $Y$. However, the two vectors are linearly dependent and so $Y$ is not a basis. 

Using a similar argument one can show that the subspace spanned by $X$ does not have any finitely supported basis, even if we allow bases that are not contained in $X$.
\end{example}

The above example shows that vector spaces with an orbit-finite basis are not closed under finitely supported subspaces.  The same issue appears with finitely supported quotients (images of surjective  linear maps). Indeed, the spaces from Examples~\ref{ex:no-equivariant-basis} and~\ref{ex:no-orbit-finite-basis} are images of equivariant linear maps from $\lin(\atoms^{(2)})$ to, respectively, $\lin(\atoms)$ and $\lin(\atoms^{(2)})$. These issues will become a problem in the next section, where we minimize weighted orbit-finite automata, since minimization will require taking subspaces and quotients.   
To deal with these issues, we notice that all spaces that we have considered so far, even if they lack an orbit-finite basis, have orbit-finite spanning sets. They are therefore instances of the following general notion: 

\begin{definition}[Orbit-finitely spanned vector space] Fix a field $\field$ and an atom structure.
    A \emph{vector space with atoms} is a vector space  such that the underlying set $V$ is a set with atoms, and the operations
    \begin{align*}
    \myunderbrace{+ : V \times V \to V}{binary addition} \qquad 
    \myunderbrace{\cdot : \field \times V \to V}{scalar multiplication}
    \end{align*}
    are finitely supported.  A vector space with atoms is called \emph{orbit-finitely spanned} if there is an orbit-finite subset of the vector space which spans it.
\end{definition}

As we have seen in Example~\ref{ex:no-orbit-finite-basis}, some orbit-finitely spanned vector spaces do not have any orbit-finite basis.
They do, however, have finite length:

\begin{lemma}\label{lem:finite-length-for-ofs}
For atoms with the finite length property,  equivariant orbit-finitely spanned vector spaces have finite length.
\end{lemma}
\begin{proof}
Let $V$ be spanned by an equivariant orbit-finite set $Q$, and let
\[ f \colon \lin Q \to V \]
be the unique linear map which extends the inclusion of $Q$ in~$V$.
This is a surjective finitely supported linear map, so (by Lemma~\ref{lem:SES}(iii)) the length of $V$ is not greater than the length of $\lin Q$, which is finite by the finite length property.
\end{proof}

\begin{theorem}[Closure properties]\label{thm:orbit-finitely-spanned-closure}
If the atoms have the finite length property, then
equivariant orbit-finitely spanned vector spaces  are closed under equivariant quotients, equivariant subspaces, direct sums and tensor products.
\end{theorem}
\begin{proof}
   Quotients and direct sums are immediate.  For the tensor product $U \otimes V$, suppose $U$ and $V$ are spanned by orbit-finite sets $Q$ and $P$.
   Since  $\lin Q \otimes \lin P$ is isomorphic to $\lin (Q \times P)$, it follows that  $U \otimes V$ is spanned by $Q \times P$, which is orbit-finite. 
   
   We are left with the subspaces. Suppose that $V$ is equivariant and orbit-finitely spanned, and let $U$ be an equivariant subspace of $V$. Construct an increasing chain of equivariant subspaces of $U$ \begin{align*}
    U_0\subseteq U_1 \subseteq U_2 \subseteq \cdots \subseteq U
    \end{align*} 
as follows:
\begin{itemize}
\item $U_0$ is the trivial space,
\item if $U_n\neq U$ then pick any vector $u_{n+1}\in U\setminus U_n$, and let $U_{n+1}$ be the subspace of $U$ spanned by all vectors in $U_n$ and all vectors in the orbit of~$u_{n+1}$.
\end{itemize}   
Each $U_n$ is an equivariant subspace of the orbit-finitely-spanned $V$, so by Lemma~\ref{lem:finite-length-for-ofs} the chain cannot grow forever, hence $U=U_n$ for some $n$. Of course, $U_n$ is spanned by all vectors in the orbits of $u_1,\ldots,u_n$.
\end{proof}

Thanks to Theorem~\ref{cor:finsup-length}, in the equality and ordered atoms,  Lemma~\ref{lem:finite-length-for-ofs} and Theorem~\ref{thm:orbit-finitely-spanned-closure} hold also for orbit-finitely spanned vector spaces that are  not necessarily equivariant.

In the remainder of this section, we discuss another closure property of orbit-finitely spanned vector spaces: if $V$ is an equivariant orbit-finitely spanned vector space, then the same is true for its \emph{finitely supported dual}, which is the vector space of finitely supported linear maps from $V$ to $\field$. We will prove closure under finitely supported duals for the equality and ordered atoms.

Before considering finitely supported duals, we discuss a slightly simpler vector space.
For~$Q$ an orbit-finite set, define 
\begin{align*}
\forms Q
\end{align*}
to be the set of finitely supported functions from $Q$ to the field $\field$. This set can be viewed as a vector space, with coordinate-wise addition and scalar multiplication.  We have
\begin{align*}
\lin Q\   \subseteq \   \forms Q,
\end{align*}
and the inclusion is strict when $Q$ is infinite, as witnessed by the function that returns $1$ on all arguments.

\begin{example}\label{ex:forms-atoms}
    For the equality atoms,
    the space $\forms \atoms$ is spanned by the functions 
    \begin{align}\label{eq:basis-for-forms-from-atoms}
    \set{f_\atoms} \ \cup \ \set{f_a : a \in \atoms}
    \end{align}
    where $f_\atoms$ maps all atoms to $1$, while $f_a$ maps $a$ to $1$ and the remaining atoms to $0$.
    Indeed, every function $f : \atoms \to \field$ supported by an atom tuple $\bar a$ can be presented as: 
    \begin{align*}
    f \ = \ \myunderbrace{f(c)}{$c$ is some fresh \\ \scriptsize atom not in $\bar a$} \cdot f_\atoms \ + \ \sum_{a \in \bar a} (f(a) - f(c)) \cdot f_a.
    \end{align*}
   The functions from~\eqref{eq:basis-for-forms-from-atoms}  are linearly independent, hence they form an orbit-finite basis of $\forms \atoms$.
    
A similar result holds for the ordered atoms, except that there the spanning set is
  \begin{align*}\label{eq:basis-for-forms-from-atoms-ordered}
    \set{f_\atoms} \ \cup \ \set{f_a : a \in \atoms} \ \cup \ \set{f_{> a} : a \in \atoms}
    \end{align*}
    where $f_\atoms$ and $f_a$ are as before and $f_{>a}$ maps $b$ to $1$ precisely if $b > a$.
\end{example}

The above example shows that $\forms \atoms$ is  orbit-finitely spanned. This is true for every orbit-finite set, not just $\atoms$:
 
\begin{theorem}\label{thm:forms-orbit-finite} Assume the equality or the ordered atoms. If $Q$ is an orbit-finite set, then  $\forms Q$ is   orbit-finitely spanned.
\end{theorem}
\begin{proof}
To simplify notation, we assume that $Q$ is equivariant. By Theorem~\ref{thm:orbit-finitely-spanned-closure} this implies the general case anyway, because every orbit-finite set $Q$ is contained in some equivariant orbit-finite set $\bar{Q}$, hence $\forms Q$ is a subspace of $\forms {\bar{Q}}$. Furthermore, for the same reasons, it is safe to assume that $Q$ is a single-orbit set.

We need to exhibit an orbit-finite set $\Phi$ of finitely supported functions, so that every finitely supported function $f:Q\to\field$ is a linear combination of functions from $\Phi$. It is enough to show this for the case where $f$ is the characteristic function of a finitely supported set $R\subseteq Q$, since such characteristic functions are easily seen to span the space of all finitely supported functions. So consider such a set $R$, supported by some tuple $\bar a$ of atoms. $R$ is then a disjoint union of $\bar a$-orbits, and it is enough to consider the case where $R$ is a single $\bar a$-orbit.

From here, the arguments for the equality and the ordered atoms begin to differ.

For the case of {\bf ordered atoms}, we know that the single-orbit $Q$ is in equivariant bijection with $\atoms\choose k$ for some number $k$, so it is enough to deal with $Q={\atoms\choose k}$. Let $R$ be the $\bar a$-orbit of some 
\[
	\textstyle r = \{a_1,\ldots,a_k\}\in{\atoms\choose k}.
\]
Each $a_i$ belongs to some $\bar a$-orbit of atoms. This orbit is an interval, so it is equal to the orbit determined by at most two atoms from $\bar a$. (One atom is enough if $a_i$ belongs to $\bar a$, or if it is larger than (or smaller than) every atom in $\bar a$.) Picking these atoms for each $i=\{1,\ldots, k\}$, we obtain a tuple $\bar b\subseteq \bar a$ of at most $2k$ atoms, such that $R$ is equal to the $\bar b$-orbit of $r$. As a consequence, $R$ is supported by $\bar b$.

As a result, we may take $\Phi$ to be the set of characteristic functions of subsets of $Q$ supported by at most $2k$ atoms. Since $k$ is fixed for a given $Q$, this is an orbit-finite set.

The argument for the {\bf equality atoms} is only a little more complicated. Here for $\Phi$ we take the set of characteristic functions of subsets of $Q$ supported by at most $k$ atoms. 

We know that for some number $k$ there is an equivariant surjection from $\atoms^{(k)}$ to $Q$ , so it is enough to deal with $Q=\atoms^{(k)}$.
Let $R$ be the $\bar a$-orbit of some 
\[
	\textstyle r = (a_1,\ldots,a_k)\in\atoms^{(k)}.
\]
We proceed by induction on the number of atoms in $r$ that are {\em not} in $\bar a$. If this number is zero, then $R$ is the singleton of $r$ and the characteristic function of $R$ is in $\Phi$, so there is nothing to do. If there are some atoms in $r$ that are not present in $\bar a$, denote by $I\subseteq\{1,\ldots,k\}$ the set of coordinates where these atoms occur in $r$. Then an $s\in \atoms^{(k)}$ belongs to $R$ if and only if:
\begin{itemize}
\item[(i)] $s$ is equal to $r$ on all coordinates apart from those in $I$,
\item[(ii)] atoms in $s$ on coordinates from $I$ do not belong to $\bar a$.
\end{itemize}
Let $\hat R\subseteq Q$ be the set of tuples $s$ that satisfy condition (i) above; obviously $R\subseteq\hat R$.
Moreover, $\hat R$ has a support of size $k-|I|$, so $\chi_{\hat R}$, the characteristic function of $\hat R$, belongs to $\Phi$.

Let $\{r_1,\ldots, r_n\}$ be the set of all $k$-tuples in $\atoms^{(k)}$ that can be obtained from $r$ by replacing some (not necessarily all, but at least one) atoms on coordinates from $I$ by some atoms from $\bar a$. Let $R_i$ be the $\bar a$-orbit of $r_i$, for each $i$. The sets $R_i$ are pairwise disjoint, and their union is equal to the difference $\hat{R}\setminus R$, therefore there is a linear equation:
\[
	 \chi_R = \chi_{\hat R}-\sum_{i=1}^n\chi_{R_i}.
\]
By the inductive assumption, each $\chi_{R_i}$ is a linear combination of functions from $\Phi$, which completes the proof.
\end{proof}

\begin{corollary}\label{cor:dual}
    Assume the equality or ordered atoms.
    If $V$ is an  orbit-finitely spanned vector space, then the same is true for its finitely supported dual.
\end{corollary}
\begin{proof}
Recall that the finitely supported dual of $V$ is the vector space of linear maps from $V$ to $\field$.
    If $V$ is spanned by an orbit-finite set $Q$, then its finitely supported dual embeds into the space $\forms Q$ by precomposing with the inclusion of $Q$ in $V$; the latter space is orbit-finitely spanned by Theorem~\ref{thm:forms-orbit-finite}.
    The corollary follows since, by Theorem~\ref{thm:orbit-finitely-spanned-closure}, orbit-finitely spanned vector spaces are closed under subspaces.
\end{proof}

In contrast to the finite-dimensional case, orbit-finitely spanned vector spaces are not isomorphic to their finitely supported duals. Indeed, Example~\ref{ex:forms-atoms} shows that
    the finitely supported dual of $\lin \atoms$ is equivariantly isomorphic to the space $\lin (1+ \atoms)$.
    There is no linear isomorphism between them which is equivariant, or even finitely supported.

The following example shows that the restriction to equality or ordered atoms was important in 
Theorem~\ref{thm:forms-orbit-finite}. This example hints on the difficulties one may encounter when generalizing 
our theory to, say, arbitrary oligomorphic atom structures.

\begin{example}
    \label{ex:counter-example-graph-atoms}
    Consider the graph atoms~\cite[Section 7.3.1]{bojanczyk_slightly2018}, i.e.~the case when the atom structure~$\atoms$ is the Rado graph. This is an undirected graph with the key property that for all finite sets $U\subseteq S\subseteq\atoms$ there is an atom $a\in\atoms$ that has an edge to every atom in $U$ and does not have an edge to any atom in $S\setminus U$.
    
We will show that the vector space $\forms \atoms$, for $\field=\Int_2$, does not have an orbit-finite spanning set.
    
Assume to the contrary that such an orbit-finite spanning set $X\subseteq \forms\atoms$ exists. Let $k$ be a number such that every function in $X$ is supported by at most $k$ atoms; such a number exists since $X$ is orbit-finite. Fix any finite set $T$ of more than $k$ atoms. Define $\chi_T : \atoms  \to \field $ to be the characteristic function of the common neighborhood of $T$, i.e.~the function which maps an atom to 1 if it has an edge to all atoms in $T$, and otherwise maps the atom to 0. This function is supported by $T$, so it belongs to $\forms \atoms$. 
    
Since the whole space is spanned by $X$, there must be a linear equation:
    \begin{equation}\label{eq:chigraph}
    	\chi_T = \sum_{i=1}^nf_i
    \end{equation}
with $f_i\in X$. Each $f_i$ is supported by some set $S_i\subseteq\atoms$ with $|S_i|\leq k$. Denote
\[
	S = T\cup\bigcup_{i=1}^nS_i.
\]
For every subset $U\subseteq T$, pick an atom $a_U$ such that:
\begin{itemize}
\item $a_U$ is a neighbour of every atom in $U$, and
\item $a_U$ is {\em not} a neighbour of any atom in $S\setminus U$.
\end{itemize}
Such an atom exists by the universal property of the Rado graph. 

Note that $\chi_T(a_U)=1$ if and only if $U=T$. 

Now consider a function $f_i$ from~\eqref{eq:chigraph}, supported by $S_i$ as above. Pick some atom $b\in T\setminus S_i$; this is possible since $|T|>|S_i|$. Since $f_i$ is supported by $S_i$, the value of $f_i$ on any atom $a$ depends only on how $a$ is connected by edges to atoms in $S_i$. This means that, for any $U\subseteq T$, 
\[
	f_i(a_{U\setminus\{b\}}) = f_i(a_U) = f_i(a_{U\cup\{b\}})
\]
(for every $U$ exactly one of these equations is trivial). As a result, the number of $U$'s for which $f_i(a_U)=1$ is even. Summing up~\eqref{eq:chigraph} over all $U\subseteq T$ we therefore obtain:
\[
	1 = \sum_{U\subseteq T}\chi_T(a_U) \stackrel{\eqref{eq:chigraph}}{=}  \sum_{U\subseteq T}\sum_{i=1}^nf_i(a_U) 
	=  \sum_{i=1}^n\sum_{U\subseteq T}f_i(a_U) = \sum_{i=1}^n0 = 0,
\] 
a contradiction.
\end{example}

\section{Minimization}
\label{sec:minimisation}
\schutz's  original paper on weighted automata contained  a minimization procedure. We  now describe a version of that procedure in the orbit-finite setting using the theory of orbit-finitely spanned vector spaces from Section~\ref{sec:vector-spaces}.  

Consider a weighted orbit-finite automaton $\Aa$ with states $Q$. We assume that the automaton is reachable, in the following sense: every vector in $\lin Q$ is a linear combination of configurations corresponding to input words.
For a vector $v \in \lin Q$, define the weighted language of $v$ to  be the  weighted language  recognized by the automaton obtained from $\Aa$ by setting the initial map to $v$, i.e.~the initial weight of a state is its coefficient in the vector $v$. Define the \emph{syntactic congruence} $\sim$ to be the equivalence relation on $\lin Q$ which identifies two vectors if the corresponding weighted languages are equal. It is not hard to see that  $\sim$ is closed under both linear combinations and applying atom automorphisms; as a result, one can speak of an equivariant quotient vector space
$(\lin Q)/_{\sim}$.
Equivalently, this is the quotient of the vector space $\lin Q$ under the subspace which consists of vectors $v$ whose corresponding weighted language is~0 everywhere. This quotient space is orbit-finitely spanned by the equivalence classes of states in~$Q$. 

In the finite dimensional case studied by \schutz, a minimal automaton is obtained by choosing some basis for this vector space, and using it as the state space of the minimal automaton. This idea, however, will not work in the orbit-finite setting, due to the difficulties with finding a basis that were described in Examples~\ref{ex:no-equivariant-basis} and~\ref{ex:no-orbit-finite-basis}.

\begin{example}\label{ex:does-not-minimize} For the equality atoms and the field of rational numbers, consider the weighted language $L \colon \atoms^* \to \field$:
    \[
    L(w) =
    \begin{cases}
    1 & \text{if } w = a b a \text{ for some } a \neq b \in \atoms, \\
    -1 & \text{if } w = a b b \text{ for some } a \neq b \in \atoms, \\
    0 & \text{otherwise.}
    \end{cases}
    \]
    This is recognized by a weighted orbit-finite automaton where the set of states $Q$ is 
    \[  \myunderbrace{\set \perp}{initial weight 1\\
    \scriptsize final weight 0}
    \quad + \quad
      \myunderbrace{ \atoms + \atoms^{(2)}}{initial weight 0\\
    \scriptsize final weight 0}
    \quad + \quad
    \myunderbrace{\set \top}{initial weight 0\\
    \scriptsize final weight 1} \]
    and  the  transitions with nonzero weight are 
    \[ \myunderbrace{(\perp, a, a),\ (a, b, (a,b)),\ ((a,b), a, \top)}{weight 1},\ \myunderbrace{((a,b), b, \top)}{weight -1} \]
    for every $a \neq b \in \atoms$.  For the syntactic congruence $\sim$, it is not hard to see that 
    $(a,b)\sim -(b,a)$ 
    for every $a \neq b$. 
    The quotient space $\lin Q /_\sim$ is:
    \[  \ \lin(\,\set{\perp, \top} + \atoms\,) \ \oplus \ X, \]
    where $X$ is the space from Example~\ref{ex:no-orbit-finite-basis}.
   This space has no finitely supported basis.
\end{example}

This example shows that  in the orbit-finite setting  the minimization procedure can leave  the realm of weighted orbit-finite automata as defined in Definition~\ref{def:weighted-orbit-finite}.
To overcome this issue, we use an alternative model for weighted automata,  which we call \emph{orbit-finitely spanned automata}. These are deterministic automata where the state spaces are orbit-finitely spanned vector spaces and all weight functions are linear maps. 

\begin{definition}
    [Orbit-finitely spanned automaton] \label{def:weighted-finite} An orbit-finitely spanned automaton consists of:
    \begin{enumerate}
        \item an orbit-finite input alphabet $\Sigma$;
        \item an orbit-finitely spanned vector space $V$;
        \item a finitely supported transition function
        \begin{align*}
        \delta : V  \times \Sigma \ \to \ V
        \end{align*}
        such that $v \mapsto \delta(v,a)$ is a linear map for every  $a \in \Sigma$;
        \item an initial vector $v_0 \in V$;
        \item a  finitely supported linear map  $F : V \to \field$.
    \end{enumerate}
\end{definition}
An orbit-finitely spanned automaton recognizes a weighted language as expected: given an input word, it computes an element of $\field$ by starting in the initial vector, then applying the transition functions corresponding to the input letters, and finally applying the final function $F$.

\begin{theorem} \label{thm:weighted-linear} Orbit-finitely spanned automata and weighted orbit-finite automata recognize the same weighted languages.
\end{theorem}
\begin{proof}
The proof resembles an analogous construction in the finite-dimensional setting, with one important difference. If one converts a weighted orbit-finite automaton with states $Q$ to an orbit-finitely spanned automaton in the natural way, then the resulting vector space  $\lin Q$ has an orbit-finite basis. Not every orbit-finitely spanned automaton arises this way though, because we do not require the vector space to have an orbit-finite basis. A key step in the proof is that every orbit-finitely spanned automaton is equivalent to one whose state space has a basis. A detailed argument is provided in Appendix~\ref{app:linear-weighted}.
\end{proof}

One advantage of orbit-finitely spanned automata is that they can be minimized. Define a \emph{homomorphism} between orbit-finitely spanned automata $\Aa$ and $\Bb$ to be a finitely supported linear map from the state space of $\Aa$ to the state space of $\Bb$, which is consistent with the structure of the automata in the natural way, see~\cite[Sec.~6.2]{bojanczyk_slightly2018}. If $\Aa$ is an orbit-finitely spanned automaton that is reachable (i.e.~its vector space is spanned by vectors that can be reached via input words), and $\sim$ is its syntactic congruence (defined in the same way as for weighted orbit-finite automata), then there is a well defined quotient automaton $\Aa_{/\sim}$. This construction, dating back to Sch\"utzenberger's original paper~\cite{schutzenberger1961definition} and described in detail e.g.~in~\cite[Thm.~8.4]{bojanczyk_toolbox}, does not depend on $\Aa$ being orbit-finitely spanned; in fact it applies to any weighted language over any alphabet. However, if $\Aa$ is orbit-finitely spanned then so is  $\Aa_{/\sim}$. The quotient automaton admits a (surjective) homomorphism from every reachable orbit-finitely spanned automaton that recognizes the same weighted language as $\Aa$. This property uniquely defines $\Aa_{/\sim}$ up to isomorphism, and so it can be called \emph{the} minimal automaton.

\smallskip

We finish this section with a third perspective on weighted languages, this time phrased in terms of monoids.
Define an \emph{orbit-finitely spanned monoid} to be a monoid $(M, \cdot, 1)$ where the underlying set $M$ is an orbit-finitely spanned vector space, and the monoid operation is bi-linear (i.e.~linear in each of the two coordinates).
We say that a weighted language $L \colon \Sigma^* \to \field$ is recognized by such a monoid if  the diagram
\[
\xymatrix@ur{
    \Sigma^* 
    \ar[r]^h
    \ar[dr]_L &
    M 
    \ar[d]^F \\
    & \field 
}\vspace{-5ex}
\]
commutes for some finitely supported monoid homomorphism $h$ and some finitely supported linear map $F$.  

The following result is somewhat unexpected, because in the non-weighted setting, orbit-finite automata and orbit-finite monoids do not recognize the same languages~\cite[Exercise 91]{bojanczyk_slightly2018}. 
\begin{theorem}\label{thm:monoids}
    Orbit-finitely spanned  monoids recognize the same weighted languages as weighted orbit-finite automata and  orbit-finitely spanned automata.
\end{theorem}
\begin{proof}
    From an orbit-finitely spanned monoid we can easily construct an orbit-finitely spanned automaton, with the same underlying vector space. For the converse direction, starting with a weighted orbit-finite automaton with states $Q$, we build a monoid out of finitely supported functions:
    \begin{align*}
    f  : \forms {(Q \times Q)}
    \end{align*}
    such that for every $p \in Q$ there are finitely many states $q \in Q$ such that $f(p,q) \neq 0$. By Theorem~\ref{thm:forms-orbit-finite}, this vector space is orbit-finitely spanned. The monoid operation is defined by:
    \begin{align*}
    (f \cdot g) (p,q) = \sum_{r \in Q} f(p,r)\cdot g(r,q),
    \end{align*}
    with the sum being finite by the assumption on $f$. This operation is finitely supported and bi-linear. The recognizing homomorphism is built using the same construction as when converting a nondeterministic automaton into a monoid.
\end{proof}

An advantage of the monoid approach is the symmetry between reading the input word left-to-right and right-to-left. In particular, the languages recognized by orbit-finitely spanned monoids are easily seen to be closed under reversals; this is harder to see  for the remaining models. 

\section{Application to  unambiguous automata}
\label{sec:unamb}
A classical application of weighted automata is a polynomial-time algorithm for language equivalence of unambiguous finite automata, i.e., nondeterministic automata with at most one accepting run for every input word.
Two unambiguous finite automata are equivalent (i.e.~they recognize the same language)  if and only if they have the same number of accepting runs for every input word (since the number of accepting runs is zero or one).
For every nondeterministic finite automaton, one can easily construct in polynomial time a weighted finite automaton which maps every input word to the number of accepting runs of the nondeterministic finite automaton; and therefore two unambiguous finite automata are equivalent if and only if the  corresponding weighted finite automata are equivalent.

In this section, we show how this result can be lifted from finite to orbit-finite automata. Consider first the case of unambiguous orbit-finite automata which are non-guessing, in the sense that they have finitely many initial states, and for every state $q$ and input letter $\sigma$, there are finitely many transitions of the form $(q,a,p)$. As explained in  Example~\ref{ex:count-runs}, for such automata we can easily count runs using a weighted orbit-finite automaton, and thus we can solve the language equivalence problem in the same time as in Theorem~\ref{thm:equivalence-orbit-finite-equality}. However, our techniques apply also to unambiguous orbit-finite automata without the non-guessing restriction.  

\begin{theorem}\label{thm:unambiguous-automata}
    Assume that the atoms are either $(\Nat,=)$ or $(\Rat,<)$.
    The equivalence problem for equivariant\footnote{As for Theorem~\ref{thm:equivalence-orbit-finite-equality}, this theorem would also work for finitely supported automata, but the notation would become more involved.} unambiguous register  automata\footnote{We formulate the theorem using unambiguous register automata, see~\cite[Sect.~2]{MottetQ19} and not for general unambiguous orbit-finite automata, so that it can be more easily compared to existing results in the literature. However, the entire proof works for the general case.}, which are allowed to use guessing, is in \exptime, and in polynomial time when the  number of registers is fixed.
\end{theorem}

This improves on previous work~\cite{MottetQ19,MottetQ19Guessing} in that: (a) we allow unrestricted guessing; (b) we allow ordered atoms and not just equality atoms;
 and (c) we improve the previous upper bounds of  \twoexpspace for an unbounded number of registers and \expspace for a fixed number of registers. 

The rest of this section is devoted to proving the above theorem. The main observation is that an orbit-finitely spanned automaton can count accepting runs for a nondeterministic orbit-finite automaton, as stated in the following lemma. 
\begin{lemma}\label{lem:unambiguous-to-weighted} Consider an equivariant nondeterministic orbit-finite automaton $\Aa$, which has finitely many accepting runs for every input word. There is an equivariant orbit-finitely spanned automaton $\Bb$, over the field of rational numbers, which outputs for every word the number of accepting runs of $\Aa$.  Furthermore, the length of the vector space used by the automaton $\Bb$ is at most 
    \begin{align}\label{eq:length-bound}
        2^{poly(k)} \cdot n^{O(k)}. 
        \end{align}
where $k$ and $n$ are as in Theorem~\ref{thm:equivalence-orbit-finite-equality}.
\end{lemma}
\begin{proof}
    Let $\Sigma$ be the input alphabet, and let $Q$ be the state space of $\Aa$.  Without loss of generality we assume that every state can reach some accepting state; the remaining states can be eliminated from the automaton without affecting  the recognized language or the numbers of accepting runs~\cite[Cor.~9.12]{bojanczyk_slightly2018}.

    For an input word $w \in \Sigma^*$, define its \emph{configuration} 
    \begin{align*}
    [w] \in \forms Q
    \end{align*}
    to be the function which maps each state $q$ to the number of runs on $w$ that begin in an initial state and end in the state $q$. The configuration produces only finite numbers, because  there cannot be a  state $q$ that can be reached via infinitely many runs over the same input word $w$; otherwise we could append some word to $w$ and get infinitely many accepting runs.
    Define 
    \begin{align*}
     V = \spanvec \set{[w] : w \in \Sigma^*} \subseteq \forms Q
    \end{align*}
    to be the subspace of $\forms Q$ that is spanned by configurations. Although the definition of $V$ uses a spanning set that is not necessarily orbit-finite (because $\Sigma^*$ is not orbit-finite), the space~$V$ is orbit-finitely spanned, as an equivariant subspace of an orbit-finitely spanned vector space, see Theorems~\ref{thm:orbit-finitely-spanned-closure} and~\ref{thm:forms-orbit-finite}.  

    We use $V$ as the state space of a orbit-finitely spanned automaton $\Bb$. Let us first prove the length bound~\eqref{eq:length-bound}. It is enough to prove the bound for the length of $\forms Q$, since $V$ is an equivariant subspace of it.
  To this end, note that the set $Q$ can be decomposed as a disjoint union of $n$ single-orbit sets with dimension at most $k$. Since 
        \[
        \length( \formscrammed {(Q_1 \! + \! Q_2)}) =
        \length(\formscrammed {Q_1}) + \length(\formscrammed {Q_2})
        \]
        it is enough to show that every equivariant single-orbit set $P$ of atom dimension at most $k$ satisfies 
        \begin{align*}
        \length ( \forms P) \le 2^{poly(k)}.
        \end{align*}
        Since $\forms P$ embeds into $\forms {\atoms^{(k)}}$, it is enough to show that the length of the latter space is at most exponential in $k$.  This follows from the proof of Theorem~\ref{thm:forms-orbit-finite}.

    We now describe the remaining structure of the orbit-finitely spanned automaton.
     The initial state is the configuration $[\varepsilon]$, which maps initial states to $1$ and non-initial states to $0$. Let us now define the transition functions. 
    For an input letter $\sigma \in \Sigma$, define a function $\delta_\sigma : V \to V$ as follows
    \begin{align*}
     \sum_{i \in I} \alpha_i [w_i]\quad \mapsto \quad  \sum_{i \in I} \alpha_i [w_i \sigma].
    \end{align*}
    We need to justify that this is well-defined. A potential problem  is that the same element of~$V$ might have several decompositions as weighted sums of configurations, and  the  output of~$\delta_\sigma$ should not depend on the choice of decomposition.  
     Consider an element of $V$ with two decompositions: 
        \begin{align*}
            \sum_{w \in W} \alpha_w [w] = 
            \sum_{w \in W} \beta_w [w],
        \end{align*}
        for some finite set $W \subseteq \Sigma^*$ of input words, and some coefficients $\alpha_w,\beta_w\in\field$. 
        We need to show that $\delta_\sigma$ produces the same output for both decomposition, i.e.
        \begin{align}\label{eq:after-sigma}
            \sum_{w \in W} \alpha_w [w \sigma] = 
            \sum_{w \in W} \beta_w [w \sigma].
        \end{align}
        Both sides in~\eqref{eq:after-sigma} are functions from $Q$ to $\field$, so to prove the equality we need to show that both sides give the same output for every state $q\in Q$. For a fixed $q$, let $P$ be the set of states $p \in Q$ such that the automaton has a transition $(p,\sigma,q)$, and furthermore $p$ appears in the configuration of some $w \in W$ with nonzero coefficient.  An important observation is that $P$ is a finite set: because the automaton $\Aa$ has finitely many accepting runs for every input word, the set $P$ contains finitely many states for every word in $W$. For every word $w \in W$ we have 
\begin{align*}
[w\sigma](q) = \sum_{p \in P} [w](p).
\end{align*}
Therefore, to prove~\eqref{eq:after-sigma}, we need to show 
\begin{align*}
    \sum_{\substack{w \in W\\ p \in P}}  \alpha_w [w](p) = 
    \sum_{\substack{w \in W\\ p \in P}}  \beta_w [w](p).
\end{align*}
This equality is indeed true, because our assumption implies a stronger equality, namely that for every $p \in P$ we have 
\begin{align*}
    \sum_{ w \in W}  \alpha_w [w](p) = 
    \sum_{w \in W}  \beta_w [w](p).
\end{align*}

    The function $\delta_\sigma$ is clearly a linear map, and thus we can set the transition function of the automaton to be $\delta(v,\sigma) = \delta_\sigma(v)$. The automaton is defined so that after reading an input word~$w$, its state is $[w]$. The final map simply takes a configuration to the sum of all coefficients that are accepting states. This concludes the proof of Lemma~\ref{lem:unambiguous-to-weighted}.
\end{proof}

\begin{proof}[Proof of Theorem~\ref{thm:unambiguous-automata}]
Consider two unambiguous register automata $\Aa_1$ and $\Aa_2$ for which we want to decide equivalence. 
Apply Lemma~\ref{lem:unambiguous-to-weighted} to each one of them, yielding orbit-finitely spanned automata $\Bb_1$ and $\Bb_2$.
Using a product construction, we get another orbit-finitely spanned automaton $\Bb$ that outputs zero for words where the automata $\Aa_1$ and $\Aa_2$ agree, and nonzero for other words.
The length of the vector space in $\Bb$ is at most twice the length of the vector spaces in $\Bb_1$ and $\Bb_2$, and hence it is at most
    \begin{align}\label{ex:bound-on-linear-length}
        2^{poly(k)} \cdot n^{O(k)}
        \end{align}
        where $k$ is the maximal number of registers used by the automata, and $n$ is the sum of the  numbers of control states.  As in the proof of Theorem~\ref{thm:equivalence-orbit-finite-equality}, we conclude that the automata~$\Aa_1$ and $\Aa_2$ are equivalent if and only if they are equivalent using input words and runs that use a number of atoms as bounded by~\eqref{ex:bound-on-linear-length}, and the latter equivalence can be tested using \schutz's polynomial time algorithm for equivalence on weighted finite automata. 
\end{proof}

\noindent{\bf Acknowledgments.}\\
We are grateful to Arka Ghosh, S\l{}awomir Lasota and Jingjie Yang, who found significant mistakes in previous versions of this paper.

\printbibliography

\newpage
\appendix

\section{Proof of Lemma~\ref{lem:SES}}\label{app:SES}

For item (i), let $g:W\to W/V$ be the (equivariant) quotient map. For the ``$\geq$'' direction, for any chains:
\begin{align*}
V_1\subsetneq \cdots \subsetneq V_n &\subsetneq V && \text{($n$ proper inclusions)} \\
U_1\subsetneq \cdots \subsetneq U_m &\subsetneq W/V &&  \text{($m$ proper inclusions)}
\end{align*}
the chain
\[
V_1\subsetneq \cdots \subsetneq V_n \subsetneq V \subseteq g^{-1}(U_1)\subsetneq \cdots \subsetneq g^{-1}(U_m) \subsetneq W
\]
has $n+m$ proper inclusions. For the ``$\leq$'' direction, for any chain of proper inclusions:
\begin{align}\label{chain1}
	W_1\subsetneq \cdots \subsetneq W_n \subsetneq W
\end{align}
consider chains:
\begin{align}
W_1\cap V &\subseteq \cdots \subseteq W_n \cap V \subseteq V \label{chain2} \\
\overrightarrow{g}(W_0) &\subseteq \cdots \subseteq \overrightarrow{g}(W_n) \subseteq W/V. \label{chain3}
\end{align}
All these inclusions are not necessarily proper. However, if the $i$-th inclusion in the first chain is an equality:
\[
	W_i\cap V = W_{i+1}\cap V
\]
then there is some $v\in W_{i+1}\setminus W_i$ such that $v\not\in V$. Since $V$ is the kernel of $g$, the $i$-th inclusion in the second chain is then proper:
\[
	\overrightarrow{g}(W_i)\subsetneq \overrightarrow{g}(W_{i+1}),
\]
so the length of~\eqref{chain1} does not exceed the sum of lengths of~\eqref{chain2} and~\eqref{chain3}. 

For item (ii), apply item (i) to $W=\lin(P+Q)$ and $V=\lin P$, noting that then $W/V=\lin Q$.

For item (iii), assume an equivariant surjection $q:P\to Q$, let $\bar{q}:\lin P\to \lin Q$ be the unique linear extension of $q$, and apply item (i) to $W=P$ and $V=\ker(\bar q)$, noting that then $W/V=\lin Q$.

\section{Proof of Lemma~\ref{lem:vsgap}}
\label{app:vsgap}

We shall proceed by induction on $n$. In the process we will need to consider a slightly more general notion:
an {\em affine subspace} of $\atoms$ is a set of the form $A = v\oplus V$, for an atom $v\in \atoms$ and a vector subspace $V\subseteq\atoms$. This is called $n$-dimensional if $V$ is.

For any $a\in\D$, let $\atoms^a_0$ be the set of all those atoms that do not contain $a$. This, of course, is a vector subspace of $\atoms$. Similarly, $\atoms^a_1$ is the set of all those atoms that do contain $a$. This is not a vector subspace, but it is an affine subspace of $\atoms$; indeed $\atoms^a_1=v\oplus\atoms^a_0$ for any $v\in\atoms^a_1$. Obviously $\atoms=\atoms^a_0\uplus\atoms^a_1$.

\begin{lemma}\label{lem:vvcapaff}
For any $n$-dimensional affine space $A$, a generator $a\in\D$ and $i=\{0,1\}$, the intersection $\atoms^a_i\cap A$:
\begin{itemize}
\item is empty, or
\item is equal to $A$, or
\item is an $(n-1)$-dimensional affine space.
\end{itemize}
\end{lemma}
\begin{proof}
First, consider the special case where $A$ is a vector space, $A=V$. Then, for $i=0$, $\atoms^a_0$ is a vector space, hence $\atoms^a_0\cap V$ is a vector space as well, and either it is equal to $V$ or it is $(n-1)$-dimensional. For $i=1$, assume $\atoms^a_1\cap V$ is nonempty and pick some $v\in \atoms^a_1\cap V$. Then $\atoms^a_1=v\oplus\atoms^a_0$ and, since $V$ is a vector space, $V=v\oplus V$. Then:
\[
	\atoms^a_1\cap V = (v\oplus\atoms^a_0)\cap (v\oplus V) =
	v\oplus(\atoms^a_0\cap V).
\]
where the second equality is by Lemma~\ref{eq:oplus-cap}(ii).
The intersection on the right is a vector space; it cannot be equal to $V$ (because $v$ is not in it), so it must be $(n-1)$-dimensional.

Now consider the general case of an arbitrary affine space $A$. Let $A=v\oplus V$ for some atom $v$ and an $n$-dimensional vector space $V$. For $i=0$:
\begin{align*}\label{eq:capoplus}
	\atoms^a_0\cap (v\oplus V) = v\oplus(\atoms^a_{v(a)}\cap V),
\end{align*}
where $v(a)\in\{0,1\}$ is the coefficient of $a$ in $v$.
Indeed, for the left-to-right inclusion, take any $w\in V$ such that $(v\oplus w)(a)=0$ or, equivalently, $v(a)=w(a)$. Then $w\in\atoms^a_{v(a)}\cap V$. For the right-to-left inclusion, do the same backwards.

Then use the special case for the vector space $V$ and $i=v(a)$.

The argument for $i=1$ is essentially the same.
\end{proof}

\begin{proof}[Proof of Lemma~\ref{lem:vsgap}]
We prove a more general statement, where $A_1,A_2,\ldots, A_k$ are affine $n$-dimensional subspaces, and not necessarily vector subspaces. We proceed by induction on the dimension $n$. For the base case, $n=1$, observe that for every finite family $A_1,A_2,\ldots, A_k$ of finite-dimensional affine subspaces of $\atoms$, the symmetric difference $A_1\xor A_2\xor\cdots\xor A_k$ has even size. Indeed, each $A_i$ is of even size, and sets of even size are closed under taking symmetric differences.

For the induction step, consider a family of $n$-dimensional affine spaces $A_1,\ldots,A_k$, and denote
\[
	X = A_1\xor\cdots \xor A_k.
\]
Assume that $X$ is nonempty. Then, as we have seen above, is must contain at least two different atoms; pick some generator $a$ on which some atoms in $X$ differ. Then both intersections
\[
	\atoms^a_0\cap X \qquad\text{and}\qquad \atoms^a_1\cap X
\] 
are nonempty.

By Lemma~\ref{eq:cap-diff}(i):
\begin{align*}
 	\atoms^a_0\cap X &= (\atoms^a_0\cap A_1)\xor\cdots\xor(\atoms^a_0\cap A_k), \\
	\atoms^a_1\cap X &= (\atoms^a_1\cap A_1)\xor\cdots\xor(\atoms^a_1\cap A_k).
\end{align*}
By Lemma~\ref{lem:vvcapaff}, both expressions to the right are symmetric differences of $(n-1)$-dimensional affine spaces. Some care is needed when (as is allowed by Lemma~\ref{lem:vvcapaff}), $\atoms^a_i\cap A_j=A_j$. In this case, present $A_j$ as
\[
	A_j = v\oplus V,
\]
use Lemma~\ref{lem:nmxor} to represent $V$ as a symmetric difference of all its $(n-1)$-dimensional subspaces, and apply Lemma~\ref{eq:oplus-diff}(iii). 

By the inductive assumption $\atoms^a_0\cap X$ is either empty or has size at least $2^{n-1}$, and similarly for $\atoms^a_1\cap X$. But we have picked $a$ so that both these intersections are nonempty, so $X$ has size at least $2^n$ as required.
\end{proof}

\section{Equivalence of linear and weighted automata}
\label{app:linear-weighted}
We begin with the easier inclusion.

\begin{lemma}\label{lem:weighted-to-linear}
    For every weighted orbit-finite automaton $\Ww$ there is an orbit-finitely spanned automaton $\Ll$ which recognizes the same weighted language.
\end{lemma}
\begin{proof}
 If $Q$ are the states of $\Ww$, then the vector space of $\Ll$  is $\lin Q$, and the transition function
    \begin{align*}
    \delta : \lin Q \times \atoms \to \lin Q
    \end{align*}
    is defined so that for every input letter $a \in \Sigma$ and states $q,p \in Q$ (which are basis vectors of the vector space $\lin Q$), the coefficient for state $p$ in the vector $\delta(q,a)$ is the weight of the transition $(q,a,p)$.
    In other words, linear map $\delta(\_,a)$ is defined by a $Q \times Q$ matrix where the coefficient in a cell $(q,p)$ is the weight of the transition $(q,a,p)$.
    Note how condition (*) from Definition~\ref{def:weighted-orbit-finite} ensures that the $\delta(q,a)$ is a well-defined vector, in the sense that it has finitely many nonzero coordinates. The initial vector of $\Ll$ consists of the initial weights in $\Ww$, and the final function is the linear extension of the original final weight function.
\end{proof}

\begin{example}
    Recall the weighted orbit-finite automaton from Example~\ref{ex:does-not-minimize}.
    We have already described the state space of the corresponding orbit-finitely spanned automaton:
    \[ V \ = \ \blue{\lin(\,\set{\bot, \top} + \atoms\,)} \ \oplus\  \red{X} , \]
    and we will now define the rest of it.
    We define the transition function on the generators of the state space:
    \begin{align*}
    &\delta(\blue{\bot}, a) = \blue{a} \quad \delta(\blue{\top}, a) = 0 \quad \delta(\blue{a},b) = \red{(a,b) - (b,a)}
    \\ &\delta(\red{(a,b) - (b,a)}, c) = \begin{cases}
        \blue{\top} & \text{if $c=a$,} \\
        \blue{-\top} & \text{if $c=b$,} \\
        0 & \text{otherwise.}
    \end{cases}
    \end{align*}
    Formally, one has to show that $\delta$ is well-defined on $\red{X}$, i.e. it satisfies 
    \[\delta(\red{(a,b) - (b,a)}, c) = -\delta(\red{(b,a) - (a,b)}, c).\]
    Similarly for the final function, we define it for generators:
    \[ F(\blue{\bot}) = 0 \quad F(\blue{\top}) = 1 \quad F(\blue{a}) = 0 \quad F(\red{(a,b) - (b,a)}) = 0 . \]
    The initial vector is simply $v_0 = \blue{\bot}$.
    This automaton accepts the same language $L$ and is in fact minimal.
\end{example}

Call an orbit-finitely spanned automaton {\em basic} if its state vector space has a basis.
For a basic automaton an equivalent weighted orbit-finite automaton can be easily produced, by using the basis as the states.  Therefore, to complete the proof of Theorem~\ref{thm:weighted-linear}, we prove the following:
\begin{lemma}\label{lem:basic-recognize-same-languages}
    For every orbit-finitely spanned automaton, there is a basic one that  recognizes the same weighted language.
\end{lemma}
\begin{proof}
    Consider an orbit-finitely spanned automaton $\Aa$ where the state space $V$ is spanned by an  orbit-finite set $Q$. Define a \emph{polynomial orbit-finite  set}, see~\cite[Definition 1]{bojanczykSingleUseRegister2019},  to be any set which is  a finite disjoint union of sets of the form $\atoms^k$.  As for every orbit-finite set, there exists a  polynomial orbit-finite set $P$ with a surjective finitely supported function from $P$ to $Q$.
    Extend this function  to surjective a linear map
    \begin{align*}
    h : \lin P \to V.
    \end{align*}
     We will define a orbit-finitely spanned automaton $\Bb$ with the state space $\lin P$ so that $h$ becomes a homomorphism  of orbit-finitely spanned automata, that is: a finitely supported linear map between the underlying vector spaces, which is consistent with the initial states, transition functions and final functions in the expected way, see~\cite[Sec.~6.2]{bojanczyk_slightly2018}. If two orbit-finitely spanned automata are connected by a homomorphism, then they recognize the same weighted language.
     Therefore, to prove the lemma it remains to define the initial state, transition function and final function in $\Bb$ so that $h$ is a homomorphism.

      For the initial state in $\Bb$ we choose some vector that is mapped by $h$ to the initial state of~$\Aa$, and for the final function we use the composition of $h$ and the final function of $\Aa$. To define the transition function, we construct a finitely supported function $\gamma$ which makes the following diagram commute:
        \begin{align*}
        \xymatrix@C=4cm{
            P \times \Sigma
            \ar[r]^\gamma
            \ar[d]_{(h, id)}
            &
            \lin P
            \ar[d]^h
            \\
            V \times \Sigma
            \ar[r]_{\text{transition function of $\Aa$}} &
            V
        }
        \end{align*}
        To this end, consider the composition of the following relations: the function $(h, id)$, the transition function of $\Aa$, and the inverse of $h$. This is a finitely supported binary relation
        \begin{align*}
        R \subseteq (P \times \Sigma) \times \lin P
        \end{align*}
        such that every element of $P \times \Sigma$ is related with at least one element of $\lin P$ (thanks to the surjectivity of $h$).  By the Uniformization Lemma from~\cite[Lemma 20]{bojanczykSingleUseRegister2019}, there exists a finitely supported function $\gamma$ which is contained in $R$, thus proving the lemma. It is worth pointing out that the Uniformization Lemma changes supports: even if $R$ is equivariant, it could be the case that $\gamma$ needs non-empty support.

          Using linearity, $\gamma$ extends to a finitely supported function
          \begin{align*}
          \bar \gamma : \lin P \times \Sigma \to \lin P
          \end{align*}
          which is a linear map for every fixed input letter; the resulting function can then be used as the transition function for $\Bb$. The commuting diagram above ensures that $h$ is a homomorphism of automata.
\end{proof}

\end{document}